%% file: single-col.tex
\documentclass[11pt,letterpaper]{article}
 \usepackage[margin=1in]{geometry}
\usepackage{graphicx}
\usepackage{multirow}
\usepackage{xcolor,xspace}
\usepackage{float,subfigure}
\usepackage{tabularx}
\usepackage{mdframed}
\usepackage{adjustbox}
\usepackage{xcolor}

\usepackage{hyperref}

\usepackage[sort]{cite}

\usepackage{colortbl}
\usepackage{paralist}
\usepackage{enumitem}
\usepackage{booktabs}
\usepackage[override]{cmtt}
\usepackage{color,xcolor}
\usepackage{graphicx,color}
\usepackage[ruled,linesnumbered]{algorithm2e}

\definecolor{mygray}{gray}{.9}
\newcommand{\gf}[1]{{\color{red} GF: #1}}

\newcommand{\mcal}[1]{\ensuremath{\mathcal {#1}}}

\newcommand{\algM}{{\ensuremath{\mcal{M}}}\xspace}

\newcommand{\algR}{{\color{darkred}\ensuremath{\mcal{R}}}\xspace}
\newcommand{\algA}{{\color{black}\ensuremath{\mathbf{A}}}\xspace}

\newcommand{\DSC}{\mathtt{DSC}}

\definecolor{goldenpoppy}{rgb}{0.99, 0.76, 0.0}
\definecolor{darkgreen}{rgb}{0,0.5,0}
\definecolor{lightblue}{RGB}{0,176,240}
\definecolor{darkblue}{RGB}{0,112,192}
\definecolor{lightpurple}{RGB}{124, 66, 168}
\definecolor{grey}{RGB}{139, 137, 137}
\definecolor{maroon}{RGB}{178, 34, 34}
\definecolor{green}{RGB}{34, 139, 34}
\definecolor{types}{RGB}{72, 61, 139}
\definecolor{gold}{rgb}{0.8, 0.33, 0.0}

\definecolor{darkgray}{gray}{0.3}

\usepackage{boxedminipage}
\usepackage{url}
\usepackage{graphicx}
\usepackage{color}
\usepackage{xspace}
\usepackage{amsmath,amsthm,amstext,amssymb,amsfonts,latexsym}

\usepackage[nameinlink,capitalize]{cleveref}
\crefname{section}{\S}{\S}
\Crefname{section}{\S}{\S}
\crefname{appendix}{App.}{Apps.}
\Crefname{appendix}{App.}{Apps.}
\crefname{theorem}{Thm.}{Thms.}
\Crefname{theorem}{Thm.}{Thms.}
\crefname{proposition}{Prop.}{Props.}
\Crefname{proposition}{Prop.}{Props.}
\crefname{definition}{Def.}{Defs.}
\Crefname{definition}{Def.}{Defs.}

\definecolor{darkred}{rgb}{0.5, 0, 0}
\definecolor{darkgreen}{rgb}{0, 0.5, 0}
\definecolor{darkblue}{rgb}{0,0,0.5}

\newcommand\markx[2]{}

\newcommand{\negl}{{\sf negl}}

\newcommand{\I}{\ensuremath{{\bf I}}\xspace}

\newcommand{\get}{\leftarrow}

\newcommand{\ignore}[1]{}

\renewcommand{\paragraph}[1]{\vspace{5pt}\noindent\textbf{#1}}

\newcounter{task}

\newtheorem{theorem}{Theorem}

\newtheorem{corollary}[theorem]{Corollary}
\newtheorem{lemma}[theorem]{Lemma}

\newtheorem{definition}[theorem]{Definition}
{
\theoremstyle{definition}
\newtheorem{remark}[theorem]{Remark}
}
\newtheorem{claim}[theorem]{Claim}

\newcounter{cnt:challenge}

\newcommand{\commentout}[1]{}

\newcommand{\view}{\textsf{view}}

\newcommand{\E}{{\ensuremath{\mathbf{E}}}\xspace}

\newcommand{\eps}{\epsilon}

\newcommand{\mz}[1]{\textcolor{red}{Mingxun: #1}}

\renewcommand{\v}{\mathbf{v}}
\newcommand{\rr}{\mathtt{RR}}

\newcommand{\bz}{\mathbf{z}}
\newcommand{\cP}{\mathcal{P}}
\newcommand{\cA}{\mathcal{A}}
\newcommand{\cB}{\mathcal{B}}
\newcommand{\cE}{\mathcal{E}}
\newcommand{\cF}{\mathcal{F}}
\newcommand{\cS}{\mathcal{S}}
\newcommand{\cZ}{\mathcal{Z}}
\newcommand{\cG}{\mathcal{G}}
\newcommand{\mmer}{\text{MinMaxErr}}
\newcommand{\emmer}{\text{EmpMinMaxErr}}
\newcommand{\err}{\text{Err}}
\newcommand{\linf}[1]{\Vert #1 \Vert_\infty}
\newcommand{\Weta}{\mathbf{W}^{(\eta)}}
\patchcmd{\maketitle}{\@fnsymbol}{\@alph}{}{}

\title{Locally Differentially Private Sparse Vector Aggregation}

\author{
    Mingxun Zhou\thanks{Carnegie Mellon University, \texttt{mingxunz@andrew.cmu.edu}}
    \and
    Tianhao Wang\thanks{Carnegie Mellon University, \texttt{tianhao@cmu.edu}}
    \and
    T-H. Hubert Chan\thanks{The University of Hong Kong, \texttt{hubert@cs.hku.hk}}
    \and
    Giulia Fanti\thanks{Carnegie Mellon University, \texttt{gfanti@andrew.cmu.edu}}
    \and
    Elaine Shi\thanks{Carnegie Mellon University, \texttt{runting@cs.cmu.edu}}
}

\date{}

\newcommand{\Q}{\mathcal{Q}}

\newcommand{\Lap}{{\sf Lap}}
\newcommand{\subE}{{\sf subE}}

\newcommand{\tB}{\tilde{B}}

\renewcommand{\I}{\mathbb{I}}

\newcommand{\svme}{\ensuremath{{\sf SVME}}\xspace}

\newcommand{\says}[2]{\noindent\textcolor{orange}{\textbf{#1 says: }}\textcolor{purple}{#2}\xspace}
\newcommand{\tw}[1]{\says{tianhao}{#1}}
\allowdisplaybreaks

\begin{document}
\maketitle

\begin{abstract}
	\input{tex/abstract}

\end{abstract}


\input{tex/intro-new}

\input{tex/roadmap}

\input{tex/related-compressed}

\input{tex/prelim}
\input{tex/definition}

\input{tex/sparse_vec}
\input{tex/eval}

\input{tex/lower}


\section*{Acknowledgments}
This work is in part supported by a Packard Fellowship, NSF awards
under the grant numbers 2128519 and 2044679, a grant from ONR and a gift from Cisco. T-H. Hubert Chan was partially funded by the Hong Kong RGC under the grants 17200418 and 17201220.

\pagestyle{plain}

\bibliographystyle{IEEEtranS}
\bibliography{bibs/refs,bibs/crypto,bibs/bibdiffpriv,bibs/cache_obliv}

\section*{\LARGE Appendices}
\input{tex/prelim-more}

\input{tex/appendix}

\input{tex/lb-detail}

\end{document}

%% file: tex/abstract.tex
Vector mean estimation is a central primitive in federated analytics.
In vector mean estimation, 
each user $i \in [n]$ holds a real-valued vector	
$v_i\in [-1, 1]^d$, and a server wants to estimate the mean of all $n$ vectors.
Not only so, we would like to protect each individual user's privacy.
In this paper, we consider the $k$-sparse version
of the vector mean estimation problem, that is, suppose
that each user's vector has at most $k$ non-zero coordinates in its
$d$-dimensional vector, and moreover, $k \ll d$.
In practice, since the universe size $d$ can be very large (e.g., the space
of all possible URLs), we would like the per-user communication to be
{\it succinct}, i.e., independent of or 
(poly-)logarithmic in the universe size. 

In this paper, we are the first to show matching upper- and lower-bounds
for the $k$-sparse vector mean estimation problem under local differential privacy.
Specifically, we construct new mechanisms that achieve 
asymptotically optimal error as well as succinct communication, 
either under user-level-LDP or event-level-LDP.
We implement our algorithms and evaluate them on synthetic
as well as real-world datasets. Our experiments 
show
that we can often achieve one or two orders of magnitude reduction 
in error in comparison with prior works under typical choices
of parameters, while incurring insignificant communication cost.

%% file: tex/intro-new.tex
\section{Introduction}
\label{sec:intro}

Federated analytics and learning allow a cloud provider to learn
useful statistics and train machine learning models
using data aggregated from a large number of users 
(e.g., browsing history, shopping records, movie ratings).
Since many of these data types are privacy-sensitive,
a line of recent work has focused on enabling privacy-preserving
federated analytics~\cite{rappor, fanti2015building, bassily2015local, wang2017locally,bonawitz2017practical,kairouz2021distributed}. 
A central primitive in privacy-preserving federated analytics 
is called {\it vector mean estimation}. 
Suppose $n$ users each have a real-valued vector
$v_i \in [-1,+1]^d$, and the collection
of all users' vectors is called
the input configuration, henceforth denoted $\v = (v_1, \ldots, v_n) \in [-1,1]^{d \cdot n}$.
The server wants to estimate the mean of the users' vectors,
without compromising each individual user's privacy.
{\it Frequency estimation}~\cite{bassily2015local} can be viewed as a special case of vector mean estimation
where each user has a binary vector $v_i \in \{0, 1\}^d$
indicating whether the user owns each of the $d$ items in some universe,
the server wants to estimate the frequency of each item. 
Besides frequency estimation, vector mean estimation 
is a key building block in numerous applications, such as
frequent item mining~\cite{sun2014personalized}, 
key-value data aggregation~\cite{pckv}, linear regression~\cite{harmony},
federated learning model update~\cite{mcmahan2017communication}, 
(stochastic) gradient descent~\cite{shokri2015privacy,abadi2016deep,dperm}, and so on.
Many of these applications are being considered
and deployed by companies such as 
Google~\cite{rappor}, Apple~\cite{appledp}, and Microsoft~\cite{microsoftdp}.


In this context, a standard privacy notion 
is {\it local differential privacy (LDP)}~\cite{ldpdef}.
Informally,
LDP (Def. \ref{def:dp}) requires that if a single user changes its input,  
the distribution of server's view 
in the protocol changes very little. 
In other words, the transcript observed by the server  
cannot allow the server to accurately infer
any single individual's input.
Two commonly-studied notions of LDP include {\it user-level} LDP (Def. \ref{def:user-ldp})
and {\it event-level} LDP (Def. \ref{def:item-ldp}).
In event-level LDP, 
we want that the distribution of the server's
view be close under two neighboring input configurations 
$\v \in [-1,1]^{d \cdot n}$ and $\v' \in [-1, 1]^{d \cdot n}$ that differ 
in exactly one coordinate (which may correspond to a single
event for a single user).
In user-level LDP, 
we want that the distribution of the server's view 
be close under two neighboring 
input configurations $\v$ and $\v'$ that differ
in the contribution of a single user, possibly involving
all $d$ coordinates of that particular user.
Unless otherwise noted, throughout this paper, we consider
an information-theoretic notion of privacy, i.e., privacy should hold
without relying on any computational assumptions.

\paragraph{Sparsity in vector mean estimation.}
In numerous practical applications, 
each user's vector $v_i \in [-1, 1]^d$ is sparse. 
We say that a vector $v_i \in [-1, 1]^d$ 
is $k$-sparse iff at most $k$ coordinates are non-zero.
We are most interested in the case when $k \ll d$.
For example, imagine that the universe 
consists of the URLs of  
all websites in the world, and each user's vector 
denotes whether a user has visited each URL. 
Another example is from natural language processing:
imagine that the universe is all possible bags-of-words of size three,
where as a user's input contains 
all occurrences that appeared in their emails. 
In such examples, $k$ is much smaller than the universe size $d$.
Sparsity has also been leveraged as an algorithmic technique
in the (non-private) federated learning literature. For example, 
Konecny et al.~\cite{konevcny2016federated}
showed that sparsifying the gradient vectors  
could result in algorithms that significantly reduce communication
while maintaining accuracy. 

In some cases, the universe may even be too large to efficiently enumerate,
e.g., the space of all possible URLs. 
In such cases, instead of writing down an estimate of the entire
mean vector $\bar{v} := \frac{1}{n} \sum_{i =1}^n v_i$, 
we want the server to instead be able
query an estimate of $\bar{v}_x$ for any $x$ of its choice (e.g., 
the fraction of users that has visited a specific URL of interest).

Due to the prevalence of sparse vectors in real-world applications, we 
ask the following important question:
\begin{itemize}
\item[]
{\it Can we achieve locally differentially private $k$-sparse vector mean
estimation with efficient communication and small error?}
\end{itemize}
In a recent workshop on Federated Learning and Analytics
hosted by Google~\cite{flworkshop}, this was raised as an important open question 
of interest to Google. 

The most na\"ive approach
is to apply the standard randomized response mechanism~\cite{randomizedresponse, dwork2006calibrating,GaussianMech}
to each coordinate --- henceforth we refer to this 
approach as ``naive perturbation''. 
For the case of event-level LDP, naive perturbation achieves 
$\widetilde{O}(\frac{1}{\epsilon \cdot \sqrt{n}})$
$L_\infty$-error with high probability 
where $\epsilon$ is the privacy budget and $\widetilde{O}(\cdot)$ hides
(poly-)logarithmic factors.
While this simple mechanism actually achieves 
asymptotically optimal error 
in light of well-known lower bounds~\cite{bassily2015local}, 
it has a high $O(d)$
communication cost.

One interesting question is whether we can achieve succinct communication 
that is independent or only logarithmically dependent 
on the universe size $d$ for sparse vectors. 
Several prior works~\cite{rappor,cormode2018marginal,bassily2015local,2018Hadamard,Bassilynips2017,bns19,as19,wang2017locally,wang-miopt-16,cko20} 
have explored this question for {\it $1$-sparse vectors},
i.e., assuming that each user holds exactly one item out of a large 
universe of $d$ items. 
The latest results~\cite{bassily2015local,Bassilynips2017,bns19,cormode2018marginal,wang2017locally} in this line of work 
showed how to achieve asymptotically optimal error 
while paying only logarithmic bandwidth
--- note that in the $1$-sparse case, event-level
LDP is the same as user-level LDP up to constant factors.

In comparison, the more general case of $k$-sparsity is less understood,
and currently we do not have matching upper and lower bounds
for schemes with succinct (e.g., logarithmic) communication.
Although some prior schemes~\cite{harmony,pckv} achieve  
communication that is succinct  in the 
universe size $d$ under even user-level differential privacy, 
they suffer from $\Omega(\sqrt{d})$ error, thus making them 
unsuitable for our 
motivating scenarios where $d$ can be very large. 
Other works~\cite{qin2016heavy,wang2018locally}
combine sampling and  
a $1$-sparse mechanism: this approach 
achieves better asymptotic error  than \cite{harmony,pckv} while still maintaining 
succinct communication, 
but their asymptotic error is a $\sqrt{k}$ or $k$ factor away from optimal,
depending on whether we care about user- or 
event-level LDP.

\subsection{Our Contributions and Results}
We give the first locally private constructions for 
vector mean estimation 
that achieve {\it succinct communication} and {\it optimal 
error} (up to polylogarithmic factors). Our contributions  
include new upper- and lower-bounds, as well as an empirical evaluation of our algorithms.

\paragraph{Upper bounds: communication-efficient LDP mechanisms.}
We devise new schemes that satisfy $\epsilon$-LDP  
(either user-level or event-level)
with the following desirable properties:

\begin{itemize}
\item  {\it Communication efficiency:}
Our mechanisms have communication 
cost that is independent of the universe size $d$, and depends only on $k$, i.e.,
the maximum number of non-zero coordinates per user.

\item  {\it Optimal error.}
Our schemes satisfy (nearly) optimal error for any $\epsilon$-LDP mechanism. 
\end{itemize}

\ignore{
Specifically, we prove the following theorem.
We shall first state a more generalized version of the theorem  
that works for $L$-neighboring vectors 
for any choice for $L$. User-level LDP and event-level LDP
are two special cases of this general theorem, corresponding
to the cases $L= 2k$ and $L= 2$, respectively. 
}

Specifically, we prove the following theorems. Although not explicitly
stated below, all of our upper bounds 
below assume the existence of a pseudorandom function (PRF) with parameter $\lambda$;
however, the PRF is needed only for measure concentration and not needed
for privacy.

\begin{theorem}[User-level LDP]
There exists an $\epsilon$-user-level-LDP 
mechanism for the $k$-sparse mean vector estimation problem
that achieves 
$O(\log k + \lambda)$ per-client communication, and 
with probability at least $1-\beta- \negl(\lambda)$, it achieves 
$O\left(\frac{1}{\eps}\sqrt{\frac{k\log (n/\beta)\log (d/\beta)}{n}}\right)$ 
$L_\infty$-error. 
Moreover, the mechanism is non-interactive, i.e., it 
involves only a single message from each client to the server.
\label{thm:intro-user}
\end{theorem}

\begin{theorem}[Event-level LDP]
There exists a non-interactive $\epsilon$-event-level-LDP 
mechanism for the $k$-sparse mean vector estimation problem
that achieves 
$O(k\log k + \lambda)$ per-client communication, and 
with probability at least $1-\beta - \negl(\lambda)$, 
it achieves 
$O\left(\frac{1}{\eps}\sqrt{\frac{\log (d/\beta)}{n}}\right)$ $L_\infty$-error.
\label{thm:intro-event}
\end{theorem}

Bassily et al.~\cite{bassily2015local} showed that any 
event-level LDP mechanism for mean estimation (even when $k = 1$)
has to suffer from at least $\Omega(\frac{1}{\eps}\sqrt{\frac{\log d}{n}})$ error.
In light of their lower bound, our event-level 
LDP mechanism achieves optimal error.
In fact, 
it turns out that our user-level LDP mechanism 
also achieves (nearly) optimal error but to show this we will
need to prove a new lower bound 
as mentioned shortly below.

\begin{table*}[ht]
\centering 
\caption{
{\bf Comparison with prior work.}
{\normalfont The ``$k$-fold repetition of $1$-sparse'' and 
``Sampling + $1$-sparse'' schemes are 
strawman constructions explained in Section~\ref{sec:roadmap}.
}
}
\begin{tabular}{c|cc|cc}
\toprule
\multirow{2}{*}{\textbf{Name}} & \multicolumn{2}{c|}{\textbf{Event-level LDP}}   & \multicolumn{2}{c}{\textbf{User-level LDP}}   \\ 
  & \textbf{Comm. Cost} & \textbf{$L_\infty$ Error} & \textbf{Comm. Cost} & \textbf{$L_\infty$ Error} \\ \midrule
$k$-fold repetition of 1-sparse & $O(k\log d)$ & $O(\frac{\sqrt{k}}{\eps}\sqrt{\frac{\log d}{n}})$ & - & - \\ 
Sampling + 1-sparse~\cite{qin2016heavy,wang2018locally} & \multicolumn{2}{c|}{same as user-level} & $O(\log d)$ & $O(\frac{k}{\eps}\sqrt{\frac{\log d}{n}})$ \\ 
Naive Perturbation\cite{GaussianMech} & $O(d)$ & $O(\frac{1}{\eps}\sqrt{\frac{\log d}{n}})$ & $O(d)$  & $O(\frac{k}{\eps}\sqrt{\frac{\log d}{n}})$ \\ 
Harmony\cite{harmony} & \multicolumn{2}{c|}{same as user-level} & $O(\log d)$  & $O(\frac{\sqrt{d}}{\eps}\sqrt{\frac{\log d}{n}})$ \\ 
PCKV\cite{pckv} & \multicolumn{2}{c|}{same as user-level} & $O(\log d)$  & $O(\frac{\sqrt{d}}{\eps}\sqrt{\frac{\log d}{n}})$ \\ 
Ours & $O(k \log k + \lambda)$ & $O(\frac{1}{\eps}\sqrt{\frac{\log d}{n}})$ &
$O(\log k + \lambda)$  & $O(\frac{\sqrt{k\log n}}{\eps}\sqrt{\frac{\log d}{n}})$ \\ 
\midrule
Lower Bounds & - & $\Omega(\frac{1}{\eps}\sqrt{\frac{\log d}{n}})$ 
& - & $\Omega(\frac{\sqrt{k}}{\eps}\sqrt{\frac{\log (d/k)}{n}})$ \\ \bottomrule
\end{tabular}
\vspace{2ex}
\label{tab:comp1}
\end{table*}

Table~\ref{tab:comp1} compares our results with prior works, and show how we achieve 
asymptotical improvements. 
Notice that our event-level scheme consumes 
more bandwidth than the user-level scheme, partly because
the optimal error bound for event-level LDP is more stringent than for user-level LDP.
It is an open question whether we can further reduce 
the bandwidth for event-level LDP  
while still preserving optimality in error\footnote{Throughout, 
we assume that the number of queries made by the server into the estimated
mean vector is polynomially or subexponentially bounded in the  
security parameter 
(denoted $\lambda$) of the PRF, depending on whether the PRF has polynomial
or subexponential security.
In cases where the server does not query the entire universe $d$, we 
take $L_\infty$ error over the queries that are actually made.
}



\paragraph{Lower bound for $k$-sparse LDP vector mean estimation.} 
We extend the proof technique of Bassily and Smith ~\cite{bassily2015local}
and prove a new lower bound for 
any user-level LDP mechanism for vector mean aggregation.
Our new lower bound almostly tightly matches
the upper bound in Theorem~\ref{thm:intro-user} (up to a logarithmic gap in $n$), showing 
our user-level LDP upper bound achieves nearly optimal error.

\begin{theorem}[(Informal) Lower bound for user-level LDP]
Any $(\epsilon, o(\frac{\epsilon}{n\log n}))$ user-level LDP 
mechanism for vector mean aggregation  
must suffer from at least $\Omega\left(\frac{1}{\eps}\sqrt{\frac{k \log (d/k)}{n}}\right)$ error in expectation.
\label{thm:intro-lb}
\end{theorem}


\paragraph{Empirical evaluation.} 
We implemented our algorithms
and the anonymized source code can be 
found at {\tt https://github.com/DPSparseVector/dp-sparse},
and we plan to open source it upon the publication of the paper.
We evaluated our algorithms using both synthetic and real-world
datasets. 
With the synthetic dataset, we could more easily control
the parameters $k$, $d$, and $n$, and we could plot
the asymptotical behavior of our algorithms. 
In comparison with prior communication-efficient works, 
our algorithms
achieve a $5.0\times$ 
reduction in $L_\infty$ error and a $29.6\times$ reduction in mean square error
for both event- and user-level LDP, 
for a typical choice of parameters, e.g., $n = 10^5$,
$d = 10^5$,  $k = 64$ and $\epsilon = 1.0$. At the same time, our algorithms consume insignificant communication cost. The report size is smaller or up to a few times larger than a TCP/IP packet headers (20 bytes).

We also tested our algorithms on several real-world datasets.
Experiment shows a $1.8\times$ to $7.3\times$ reduction
in $L_\infty$ error and a $3.1\times$ to roughly $114.3\times$ reduction
in mean square error compared to prior schemes.

\paragraph{Additional contributions.}
Besides the commonly considered user-level and event-level LDP, 
as a by-product of our upper bound constructions, we 
come up with a communication-efficient LDP mechanism 
under a more generalized $L$-neighboring notion.
Two input configuations $\v = (v_1, \ldots, v_n) \in [-1,1]^{d \cdot n}$ and 
$\v' = (v'_1, \ldots, v'_n) \in [-1,1]^{d \cdot n}$
are said to be $L$-neighboring, iff 
they are otherwise identical except for  
one user's coordinates $v_i$ and $v'_i$, and moreover,
$\Vert v_i - v'_i\Vert_1 \leq L$.
Roughly speaking, 
a mechanism satisfies $(\epsilon, \delta)$-LDP for $L$-neighboring 
input configurations if 
the server cannot $(\epsilon, \delta)$-distinguish 
two $L$-neighboring input configurations (under the standard distance
notion of $(\epsilon, \delta)$-differential privacy).
Note that the commonly known user- and event-level  
LDP notions are special cases of the above more generalized notion, 
for $L = 2k$ and $L=2$, respectively.
Therefore, introducing the generalized $L$-neighboring notion allows
us to study user- and event-level LDP under a more unified lens;
and indeed 
we use it as an intermediate stepping stone to get our main results 
(Theorems~\ref{thm:intro-user}, \ref{thm:intro-event}, and \ref{thm:intro-lb}).
We believe that 
this generalized $L$-neighboring notion can be of independent interest 
in some practical applications. For example, 
Abadi et al.~\cite{abadi2016deep} considered a gradient clipping technique
where each user would clip its gradient vector to a smaller range 
(thus pruning excessively large or small values)
before sending it to the server. This technique allows them to 
more tightly bound the $L_1$ norm of each user's contribution.

%% file: tex/roadmap.tex
\section{Technical Roadmap}
In this section, we give an informal technical overview of our results.

\label{sec:roadmap}

\subsection{Warmup: an Event-Level LDP Mechanism for Frequency Estimation}
\label{sec:roadmap-warmup}

For simplicity, we 
first focus on 
the special case of designing a {\it frequency estimation}
mechanism that satisfies {\it event-level} LDP.
Recall that the frequency estimation problem is a special
case of our general formulation of vector mean estimation. 
In frequency estimation, each client $i \in [n]$ has a binary vector 
$v_i \in \{0, 1\}^d$, 
denoting whether the client owns each item from a universe of $d$ items.
The server wants to estimate the frequency of each item.
Once we understand how to design an event-level LDP mechanism
for frequency estimation, we can later 
extend our techniques to 1) support {\it user-level} LDP; and 2)
support the more general case of vector mean estimation
where the client's vector is from a real domain.

\paragraph{Strawman: $k$-fold repetition of the $1$-sparse scheme.}
Recall that prior works~\cite{bassily2015local, wang2017locally, cormode2018marginal} have proposed $1$-sparse frequency estimation mechanisms
that achieve optimal error, that is, $\widetilde{O}(\frac{1}{\epsilon \sqrt{n}})$ 
error, incurring only logarithmic communication.
In our problem, each client
owns $k$ items rather than $1$.
Therefore, a strawman idea 
is through a $k$-fold repetition of the $1$-sparse scheme.
Specifically, each client can pretend to be 
$k$ virtual clients, and each 
virtual client owns only one item. 
Imagine that we run a $1$-sparse scheme over these $k n$ virtual clients.  
Since each client acts as $k$ virtual clients, its communication cost
is $\widetilde{O}(k)$ which is independent of the universe size $d$.
The resulting $L_{\infty}$ error 
would be $\widetilde{O}(\frac{1}{\epsilon \sqrt{k n}})$ 
over all $k n$ virtual clients.
In reality, we want to take the mean over the $n$  
real clients. After renormalizing, the 
actual error is 
$\widetilde{O}(\frac{\sqrt{k}}{\epsilon \sqrt{n}})$.

This strawman scheme gives non-trivial bounds,  
but 
does not achieve optimal 
$\widetilde{O}(\frac{1}{\epsilon \sqrt{n}})$ error.

\paragraph{Our idea.}
We devise a new scheme that 
combines the elegant ideas behind the $1$-sparse
mechanism by Wang et al.~\cite{wang2017locally} 
with a new random binning idea.
Our approach is as follows:
\begin{itemize}[leftmargin=3mm]
\item
\underline{\it Each client} $i \in [n]$ 
does the following:
\begin{enumerate}[leftmargin=3mm]
\item Sample two random hash functions 
$h_i : [d] \rightarrow [k]$
and $s_i : [d] \rightarrow \{-1, 1\}$.
\item 
Let $\{x_1, \ldots, x_k\}\in [d]^k$
denote the $k$ items belonging to client $i \in [n]$. 
For each $j \in [k]$, place $x_j$ into the hash bin  
indexed $h_i(x_j)$. Note that in total, there are $k$ hash bins per client.
\item 
For each hash bin $j \in [k]$, compute
$B_{i, j} := \sum_{x \in {\sf bin}_j } s_i(x) + {\sf Lap}(\frac{1}{\epsilon})$
where ${\sf Lap}(\frac{1}{\epsilon})$
denotes Laplacian noise of average magnitude 
$\frac{1}{\epsilon}$.
\item 
Send to the server the tuple $(h_i, s_i, \{B_{i, j}\}_{j \in [k]})$ where $h_i$ and $s_i$ denote
the description of the two hash functions.
\end{enumerate}
\item
\underline {\it Server} does the following to estimate the fraction 
of clients that own an arbitrary item $x^* \in [d]$:
\begin{enumerate}[leftmargin=3mm]
\item 
For each client $i \in [n]$, compute $j^* = h_i(x^*)$. 
\item 
Output $\frac{1}{n}\sum_{i \in [n]} B_{i, j^*} \cdot s_i(x^*)$.
\end{enumerate}
\end{itemize}

As mentioned later, to get our desired bounds, 
we need the hash functions 
$h_i$ and $s_i$ to be pseudorandom --- however, we stress
that the pseudorandomness assumption is needed only 
for load-balancing among the hash bins 
and not for proving privacy.
In other words, our scheme satisfies information-theoretic LDP.
Specifically, 
to sample a pseudorandom function (PRF), the client
samples a random seed whose 
length is related to the strength of pseudorandomness and independent of $d$.
To send the description of the hash function to the server,
the client sends the pseudorandom seed 
to the server.

Finally, in practice, 
the client can clip each $B_{i, j}$ to an integer value between $[-\tilde{O}(k), \tilde{O}(k)]$
before sending it to the server --- this does not affect the 
privacy analysis or our asymptotic error bound.
In this case, the per-client communication of our scheme 
is at most $O(k \log k)$ 
plus the description of 
the hash function (e.g., the seed of a PRF).

\paragraph{Informal utility analysis.}
To gain intuition, we present an informal analysis of our scheme.
The formal proofs (for the more generalized vector mean estimation scheme) 
are deferred to Section~\ref{sec:sparse-vec}. 
Note that understanding the utility analysis also helps to 
understand why the scheme works.

Henceforth, we 
use ${\sf bin}^i_j$ to denote 
the set of items client $i$ places into its $j$-th bin (and when
it is clear from the context which client $i$ we are referring to,
we may omit $i$).
Let $C_{i, j} = \sum_{x \in {\sf bin}^i_j}s_i(x)$
be the true aggregated ``count''
of the $j$-th bin belonging to the $i$-th client.
Suppose that the server wants to know the 
frequency of item $x^* \in [d]$.
To do this, the server computes the summation 
$\sum_{i \in [n]} B_{i, h_i(x^*)} \cdot s_i(x^*) = 
\sum_{i \in [n]} C_{i, h_i(x^*)}
\cdot s_i(x^*) + \sum_{i \in [n]}{\sf Lap}(\frac{1}{\epsilon})$ --- note that
here we have not normalized the sum with the $\frac{1}{n}$ factor yet, we can defer
this step to the end.  
The second part of the summation $\sum_{i \in [n]}{\sf Lap}(\frac{1}{\epsilon})$,
is the summation of $n$ independent ${\sf Lap}(\frac{1}{\epsilon})$ noises,
and thus 
its magnitude is roughly $\widetilde{O}(\frac{\sqrt{n}}{\epsilon})$.
The first part  of the summation
$\sum_{i \in [n]} C_{i, h_i(x^*)}
\cdot s_i(x^*)$
can be further decomposed into two sources of contributions:
\begin{enumerate}
\item  
Each client $i$ who owns $x^*$
contributes one $+1$ term to the summation because $(s_i(x))^2=1$.
\item 
For each client $i$ and each item $x \neq x^*$ 
owned by the client 
such that $h_i(x) = h_i(x^*)$, 
it contributes $s_i(x)\cdot s_i(x^*)$ to the summation, which is a random choice of 
$-1$ or $+1$ assuming that $s_i(\cdot)$ is a random oracle.
\end{enumerate}
Thus, 1) corresponds to 
to the true count of the item $x^*$, 
whereas 2) is can be viewed as the result 
of a random walk of expected length $O(n)$, 
i.e., a random noise of magnitude roughly $\widetilde{O}(\sqrt{n})$.
In particular, the length of this random walk is upper bounded by the total
load of the hash bins $\sum_{i \in [n]} | {\sf bin}^i_{h_i(x^*)}|$, which is $O(n)$
assuming that each $h_i$ is a random oracle.

Summarizing the above, the estimated count  
is the true count plus roughly $\widetilde{O}(\frac{\sqrt{n}}{\epsilon})$ noise.
Finally, when the server normalizes
the above sum by $\frac{1}{n}$ to compute the average, the resulting error becomes 
$\widetilde{O}(\frac{1}{\epsilon \sqrt{n}})$.
Note that in Theorem~\ref{thm:intro-event}, the 
precise expression for the error bound has an extra $\log\frac{d}{\beta}$ term
which we ignore here, 
where $\beta$ is the failure probability for the error bound.   
Specifically, the $\log d$ term arises from taking a union bound 
over the universe of $d$ elements and the $\log \frac1\beta$ 
term comes from a precise measure concentration bound on the error ---
we defer these precise calculations to the subsequent technical sections.

At this point, it is helpful to observe that 
in this construction, the error comes from two sources --- this observation
will later help us to generalize the scheme to user-level LDP:
\begin{itemize}[leftmargin=5mm]
\item 
{\it Noise component:} the first source of error is the  
sum of $n$ independent ${\sf Lap}(\frac{1}{\epsilon})$ noises, 
one for each ${\sf bin}^i_{h_i(x^*)}$
where $i \in [n]$;
\item 
{\it Colliding items component:}
the second source of error is 
the random contribution of either $+1$ or $-1$ from each element 
$x \neq x^*$ that each client $i$ 
places into its bin 
${\sf bin}^i_{h_i(x^*)}$.
\end{itemize}

\begin{remark}
In the above, we assumed that the hash functions $h_i$'s and $s_i$'s
are random oracles. 
In practice, we instantiate the hash functions using pseudorandom functions.
\end{remark}

\paragraph{Informal privacy analysis.}
We now give an informal privacy analysis, while deferring
the formal proofs to Section~\ref{sec:sparse-vec}.
We want to show that the scheme satisfies $\epsilon$-event-level-LDP.
Fix the hash functions $h_1, \ldots, h_n, s_1, \ldots, s_n$, 
and consider two input configurations $\v, \v' \in \{0, 1\}^{d \cdot n}$
that differ in only one position. 
Let $C_{i, j} = \sum_{x \in {\sf bin}_{i, j}} s_i(x)$
be the true aggregated ``count''
of the $j$-th bin belonging to the $i$-th client, when the input configuration
is $\v$; and let $C'_{i, j}$ be the corresponding 
quantity when the input configuration is $\v'$.
It must be that all $C_{i, j}$
and $C'_{i, j}$ are 
the same everywhere except for one bin $j^*$ corresponding to one client $i^*$.
Moreover, for the only location where they differ, 
it must be that $|C_{i^*, j^*} - C'_{i^*, j^*}| \leq 1$.
Having observed this, it is not too hard to show 
that adding  Laplacian noise of average
 magnitude $\frac{1}{\epsilon}$ to each bin suffices
for achieving $\epsilon$-event-level-LDP.

\subsection{Extension: a User-Level LDP Mechanism for Frequency Estimation}
\label{sec:roadmap-user}
One trivial way to obtain user-level LDP is to 
directly use the aforementioned warmup 
scheme, and simply apply standard privacy composition theorems~\cite{dworkrothdpbook,kairouz2015composition} 
to reset the parameters.
Specifically, 
to achieve $\epsilon$-user-level-LDP, we would need to plug in  
a privacy  parameter of $\frac{\epsilon}{k}$ when invoking
the warmup scheme.
This results in an error bound of 
$\widetilde{O}(\frac{k}{\epsilon \sqrt{n}})$ which is an $\widetilde{O}(\sqrt{k})$
factor away from optimal.

\paragraph{Strawman: sampling + $1$-sparse mechanism.}
Another strawman idea for each client to randomly sample $1$ item out of its $k$ items,  
apply the $1$-sparse mechanism to the sampled items,   
and finally, renormalize the estimate accordingly~\cite{qin2016heavy,wang2018locally}. 
Unfortunately, it is not hard to show that the  
resulting error would again be 
$\widetilde{O}(\frac{k}{\epsilon \sqrt{n}})$, an $\widetilde{O}(\sqrt{k})$
factor away from optimal.
Note also that if a client has strictly fewer than $k$ items, it needs
to first pad its input to $k$ with filler items, and then apply the the sampling
and $1$-sparse mechanism.

\paragraph{Our approach.}
Our approach is to generalize our warmup mechanism.
Suppose we want to achieve $(\epsilon, \delta)$-LDP 
under $L$-neighboring input configurations. 
Recall that two input configurations 
$\v = (v_1, \ldots, v_n) \in \{0, 1\}^{d\cdot n}$ and
$\v' \in (v'_1, \ldots, v'_n)  \in \{0, 1\}^{d\cdot n}$
are $L$-neighboring iff they differ in only one user's
contribution $v_i$ and $v'_i$, and moreover, 
$|v_i - v'_i|_1 \leq L$.
Note that user-level LDP
is simply a special case where $L = k$.
In other words, we want the server's view to be $(\epsilon, \delta)$-close
for two input configurations 
$\v = (v_1, \ldots, v_n) \in \{0, 1\}^{d\cdot n}$ and
 $\v' \in (v'_1, \ldots, v'_n)  \in \{0, 1\}^{d\cdot n}$
under the distance notion 
of the standard $(\epsilon, \delta)$-differential privacy definition~\cite{dwork2006calibrating}.
In our reasoning below, we will carry around the parameter $L$, 
and at the end, we can plug in $L = k$ to get the user-level LDP result.
However, as noted earlier in Section~\ref{sec:intro},
the more general scheme parametrized by $L$
can be of independent interest.

Our generalized scheme is almost the same as the warmup scheme, except
with the following modifications:
\begin{itemize}
\item 
Each client now has $b$ hash bins
rather than $k$ bins as in the warmup scheme.
For now, we leave the choice of $b$ unspecified,
and work out the optimal choice later.
\item 
Each client $i \in [n]$ now computes the noisy sum $B_{i, j}$ as 
$B_{i, j} := \sum_{x \in {\sf bin}_j } s_i(x) + {\sf Lap}(\frac{1}{\epsilon'})$
where $$\epsilon' = \frac{O(\epsilon)}{\min(L, \sqrt{b L \log\frac{b}{\delta}})}.$$
\end{itemize}

\paragraph{Informal privacy analysis.}
Consider two $L$-neighboring input configurations 
$\v$ and $\v'$, and fix all hash functions $h_1, \ldots, h_n$
and $s_1, \ldots, s_n$. 
Let $C_{i, j} := \sum_{x \in {\sf bin}^i_j } s_i(x)$
be the true ``count'' of ${\sf bin}^i_j$ 
under $\v$ and let 
$C'_{i, j}$ be the corresponding quantity under $\v'$.
Now, consider the 
vectors ${\bf C} = \{C_{i, j}\}_{i \in [n], j \in [b]}$  
and 
${\bf C}' = \{C'_{i, j}\}_{i \in [n], j \in [b]}$.
We want to show that $|{\bf C} - {\bf C}'|_1 \leq \min(L, O(\sqrt{{b L\cdot \log\frac{b}{\delta}}})$ with probability $1-\delta$. 
If so, adding the aforementioned noise is sufficient  
for achieving $(\epsilon, \delta)$-LDP under $L$-neighboring.
Now, $|{\bf C} - {\bf C}'|_1 \leq L$ is easy to see. 
Therefore, it suffices to show that
$|{\bf C} - {\bf C}'|_1 \leq O(\sqrt{{b L\cdot \log\frac{b}{\delta}}})$ with probability
$1-\delta$.
Due to standard measure concentration bounds, when we change $\v$ to $\v'$,  
for any fixed ${\sf bin}^i_j$, 
it holds that $|C_{i, j} - C'_{i, j}| \leq 
O(\sqrt{\frac{L\cdot \log\frac{b}{\delta}}{b}})$ with probability $1-\frac{\delta}{b}$.
Taking a union bound over all $b$ bins, we have that
$|{\bf C} - {\bf C}'|_1 \leq 
O(\sqrt{{b L\cdot \log\frac{b}{\delta}}})$ with probability $1-\delta$.


\paragraph{Informal utility analysis and optimal choice of $b$.}
As in the earlier event-level LDP scheme, 
the error in the final summation 
$\sum_{i \in [n]} B_{i, h_i(x^*)} \cdot s_i(x^*)$ --- without normalizing
it with the $1/n$ factor yet
--- comes from two sources:
\begin{itemize}
\item {\it Noise component.}
The noise component 
consists of the summation 
of $n$ independent ${\sf Lap}(\frac{1}{\epsilon'})$
noises where 
$\epsilon' = \frac{O(\epsilon)}{\min(L, \sqrt{b L \log\frac{b}{\delta}})}$.
Thus, the total noise is roughly 
$O(\frac{1}{\epsilon}\cdot \sqrt{n} \cdot \min(L, \sqrt{b L \log\frac{b}{\delta}}))$.

\item {\it Colliding items component.}
The contribution from all colliding elements 
can be viewed as a random walk of length that is equal to the number
of colliding elements  
in all $n$ bins $\{{\sf bin}^i_{h_i(x^*)}\}_{i \in [n]}$.
The number of colliding elements 
is concentrated around its expectation 
$\frac{nk}{b}$ with high probability, and thus
the colliding items component results  
in roughly $\widetilde{O}(\sqrt{\frac{nk}{b}})$ error.
\end{itemize}

The total error is minimized 
when the noise component is roughly equal to the contribution from colliding
elements, and we derive the optimal choice of $b$ as  
\begin{equation}
b = \begin{cases}
\tilde{\Theta}\left( \frac{k}{L^2} \right), \textit{ if } L \le k^{\frac{1}{3}} \\
\tilde{\Theta}\left( \sqrt{\frac{k}{L}} \right), \textit{ otherwise. } \\
\end{cases}
\label{eqn:numbin}
\end{equation} 
When the above optimal $b$ is chosen correspondingly, both error components 
are roughly equal (omitting the logarithmic factors).  Specifically, for the 
the case when 
$L \le k^{\frac{1}{3}}$, both error components are roughly 
$\tilde{O}\left(\frac{1}{\eps}L\sqrt{n}\right)$; for the other case,
both error components are 
$\tilde{O}\left(\frac{1}{\eps}(kL)^{\frac{1}{4}}\sqrt{n}\right)$.
Keep in mind that for our final error bound, we need to apply an extra $1/n$
normalizing factor to the above terms.

Summarizing the above, 
we obtain a mechanism that satisfies $(\epsilon,\delta)$-LDP under $L$-neighboring,
with per-client communication cost 
$\widetilde{O}(b)$ where $b$ is shown in Equation~(\ref{eqn:numbin}),
and its choice depends on whether $L \leq k^{\frac13}$.
Further, the scheme achieves 
$\tilde{O}\left(\frac{1}{\eps}\min(L,(kL)^{\frac{1}{4}})\sqrt{\frac{1}{n}}\right)$ error 
after applying the extra $1/n$ normalization factor.

For the special case of user-level-LDP, which can be captured by $k$-neighboring LDP,
using the above calculation, we conclude that the optimal choice of $b$ should be $b = 1$. 
In this case, the error is roughly  
$\tilde{O}(\frac{\sqrt{k}}{\eps \sqrt{n}})$.
So far, the scheme described above achieves $(\epsilon, \delta)$-LDP with
a non-zero $\delta$. However, for the special case of user-level LDP,
we can use an additional clipping technique 
to obtain $\epsilon$-LDP. We defer the detailed exposition of this
technique to 
subsequent sections.

Finally, observe that 
for $L=1$, i.e., for event-level-LDP, 
the optimal choice of $b = O(k)$.  
This shows that our event-level-LDP scheme 
in Section~\ref{sec:roadmap-warmup}
is also a special case of 
the above more generalized scheme.

\ignore{
\mz{I added the following description. }

\mz{There are mulitple ways to describe the
techniques.
One way is to describe the algorithm for 
$L$-LDP first, then describe some 
optimizations in user-level LDP
that can achieve pure DP. 
The other way is to directly 
describe the pure-DP algorithm for 
user-level LDP and 
defer the description for general
$L$ to the extension. In 
Section~\ref{sec:sparse-vec},
we describe the algorithm using the 
latter one.}

Moreover, under the conventional user-level LDP setting where $L=k$, 
the optimal bin choice is simply $b=1$ 
and therefore the absolute error is
$\tilde{O}\left(\frac{1}{\eps}\sqrt{\frac{k}{n}}\right)$ 
with high probability. 

We further improve the algorithm to pure-LDP 
by clipping and get similar asymptotic utility
guarantee. 
The idea is to let the failure probability $\beta$ 
in the utility theorem absorb the 
failure probability $\delta$
in the $(\eps,\delta)$-DP definition.
The improved algorithm clips the value in the 
single bin to range $[-\eta,\eta]$, where
$\eta=O(\sqrt{k\log (n/\beta)})$.
Then, adding Laplacian noise 
${\sf Lap}(\frac{2\eta}{\eps})$ 
is enough to ensure 
$\eps$-user-level-LDP.
The overflow probability is bounded by
$O(\beta)$
and the rigorous analysis gives
$\tilde{O}(\frac{1}{\eps}\sqrt{\frac{k}{n}})$
absolute error 
with high probability.
}

\subsection{Generalizing to Real-Valued Vectors}
In the more general case, each client $i$ holds 
a real-valued vector $v_i \in [-1, 1]^d$, with at most $k \ll d$ 
non-zero coordinates.
For example, each non-zero coordinate may represent
the rating a user has given to a movie 
that it has watched. Chances are, each user has watched 
relatively few ($k$) movies 
out of the entire universe of $d$ movies. 

It is not difficult to generalize the aforementioned 
schemes (Sections~\ref{sec:roadmap-warmup} and \ref{sec:roadmap-user})
to real-valued vectors. 
The only modification is the following: each client $i$ now computes 
$B_{i, j}$ as follows for $j \in [b]$ where $b$ denotes the number of bins per client:
\[
B_{i, j} := \sum_{x \in [d]} s_i(x) \cdot v_{i,x} + {\sf Lap}\left(\frac{1}{\epsilon'}\right)
\]
where $v_{i,x}$ denotes the $x$-th coordinate of the client's vector $v_i$,
and the choice of $\epsilon'$ is the same as before.
Note that since our event-level-LDP scheme (Section~\ref{sec:roadmap-warmup})
is a special case of the scheme in Section~\ref{sec:roadmap-user}, the above
works for the event-level-LDP scheme too. 

The proof of the above generalized 
scheme is similar in spirit to the binary case but requires more careful
calculation. 
In the subsequent technical sections, we directly prove the real-valued
case, since this is the more general form.

\subsection{Our Lower Bound}
The framework in Bassily and Smith~\cite{bassily2015local} provides a lower bound for the event-level LDP under 1-sparse setting. 
Our event-level LDP upper bound tightly matches the lower bounds and therefore 
closes the case for event-level LDP.
We observe that it is not hard to extend 
Bassily and Smith~\cite{bassily2015local}'s proof to 
user-level-LDP, and the resulting lower bound  
matches the error achieved by our earlier upper bound.
We defer the detailed presentation of the lower bound to Section~\ref{sec:lower}.


%% file: tex/related-compressed.tex
\subsection{Additional Related Work}
\label{sec:related}
\paragraph{Frequency estimation under LDP}.
Privacy-preserving frequency estimation is a fundamental
primitive in federated analytics.
Earlier works in this space  
focused on the case when the universe size $d$ is small, and these works
often suffer from per-client communication 
cost proportional to $d$.
For example, 
RAPPOR and its variants~\cite{rappor,wang-miopt-16,optimalschemes2018Ye} 
encode each client's item 
with one-hot encoding and performs coordinate-wise randomized response (RR)~\cite{randomizedresponse}, which suffers from at least $d$ communication cost.
Various subsequent works~\cite{bassily2015local,fanti2015building,2018Hadamard,Bassilynips2017,bns19,as19,wang2017locally,wang-miopt-16,cko20} 
focused on how to compress the communication 
especially  when the universe size $d$ is large, 
but each client has only one non-zero coordinate
(i.e., the $1$-sparse case).
Some of these algorithms~\cite{bassily2015local,2018Hadamard,Bassilynips2017,bns19,cko20,wang2017locally} achieved optimal estimation error 
and using only logarithmic bandwidth. 

When each user can have up to $k$ items, one
approach is to ask users to sample one item to report (e.g., \cite{qin2016heavy,wang2018locally}) using the $1$-sparse protocol (reviewed above) as a black-box.  This approach introduces an error that is a $\sqrt{k}$ factor away from optimal for user-level LDP, and a $k$ factor away from
optimal for event-level LDP. 


\paragraph{Vector mean estimation under LDP.}
For vector mean estimation under LDP, 
a few earlier works~\cite{duchi2018minimax,DuchR19}.
showed how to achieve optimal error 
for the dense case when $d \approx k$, absent communication constraints. 
Bhowmick et al.~\cite{bhowmick2019protection} showed how to achieve asymptotically optimal
accuracy when $\epsilon > 1$, but they require $\Omega(d)$ communication.
Following works, such Harmony~\cite{harmony}, Wang et al.~\cite{WangNing2019}, Li et al.~\cite{li2020numeric} and Zhao et al. ~\cite{zhao2020local} improve the utility 
compared to \cite{DuchR19}. However, 
all of the above works focused on the dense case and did not consider sparse vectors.
Chen et al.~\cite{cko20} achieved  
optimal error and succinct communication
for the $1$-sparse case. 
The PrivKVM work~\cite{privkv} 
proposed an {\it interactive} protocol for vector mean estimation but 
it suffers from at least $\sqrt{d}$ error; the approach was later  
improved~\cite{privkvm} but the protocol is still interactive.

\paragraph{Computational differential privacy.}
Our work focuses on an information theoretic notion of privacy.
An orthogonal line of work considered computational differential privacy (CDP)~\cite{cdp}
in distributed analytics~\cite{ndss11,fc12,bonawitz2017practical,deeplearncdp}.
Some of these works 
showed how to compute distributed summation with error
comparable to central DP, relying on cryptographic assumptions.
Recently, Bagdasaryan \emph{et al}~\cite{bagdasaryan2021towards} 
considered frequency estimation  
under CDP assuming $1$-sparsity, with the extra assumption that the frequency  
vector must be sparse too. 
For the more general setting of $k$-sparsity that we consider, 
it is not known how CDP can further improve the acccuracy 
in comparison with LDP, 
while still preserving succinct communication.
We leave this as an open 
question.
 
\paragraph{Sparse vector data releasing under central-DP.}
Previous works \cite{cormode2012differentially, korolova2009releasing, aumuller2021differentially} discussed a related setting that a single entity wishes to differentially privately release a $k$-sparse vector $v\in[0,u]^d$ ($u$ can be large). The neighboring notion is also defined by $L_1$ distance -- a neighboring input pair $v\sim v'$ iff. $\Vert v-v'\Vert_1 \le 1$. For example, the newest work on this line -- the ALP mechanism~\cite{aumuller2021differentially} showed how to privately encode the $k$-sparse vector with $O(k \log(d+u))$ bits with $L_\infty$ decoding error of $O(\frac{\log d}{\eps})$. However, the encoding-decoding processes of these works are biased. Although this is acceptable in one-time data releasing, it is not suitable for mean estimation because the biased error will add up $n$ times. Therfore, there is no concentration property on the final estimation error. We implemented the ALP mechanism under event-level LDP and the mean estimation error is much worse than the simple $k$-fold repetition scheme. It is unclear how to debias these schemes to fit the need of mean estimation.

%% file: tex/prelim.tex
\section{Preliminaries and Definitions}
\subsection{Background on Differential Privacy}
Differnetial privacy was first proposed by Dwork et al.~\cite{dwork2006calibrating}.
and has since become a {\it de facto} privacy notion.

\begin{definition}[$(\eps,\delta)$-close]
	We say the distributions of two random variables, $X$ and $X'$ are $(\eps,\delta)$-close iff they have the same domain $D$ and for every subset $S\subseteq D$, 
	
	$$\Pr[X\in S] \le e^\eps \Pr[X'\in S]+\delta.$$
\end{definition}

\begin{definition}[$(\eps,\delta)$-Differential Privacy]
	A function $f$ is $(
	\eps,\delta)$-DP w.r.t. some neighboring relation $\sim$ on its input domain iff for every pair $v,v' \in \text{Domain}(f)$, s.t. $v\sim v'$, the distributions of $f(v)$ and $f(v')$ are $(\eps,\delta)$-close.
\end{definition}
If a function $f$ is $(
        \eps,0)$-DP, 
 we also 
say that $f$ is $\epsilon$-DP for short (w.r.t. the neighboring relation $\sim$).

%% file: tex/definition.tex
\subsection{Sparse Vector Mean Estimation}
Consider $n$ clients, indexed by the set $[n]=\{1,2,\dots,n\}$. Each client has a real-value vector $v_i \in [-1,1]^d$. Also, each vector $v_i$ is $k$-sparse, i.e., 
it has at most $k$ non-zero coordinates. Different clients
may have different non-zero coordinates.
We use the notation $\v:=(v_1,\dots,v_n)$ to denote all clients' inputs, and we also refer
to $\v$ as an input configuration. 
A server wants to estimate the mean vector, $\bar{v}=\frac{1}{n}\sum_{i\in[n]}v_i$ 
through a non-interactive mechanism. 

In a non-interactive mechanism, each client 
sends a \textbf{single} message to the server, and 
the server then computes an estimate of the mean vector
$\bar{v}=\frac{1}{n}\sum_{i\in[n]}v_i$.
Both the clients and the server can make use of randomness
in their computation.

Henceforth, let $\sim$ denote some symmetric neighboring
relation defined over two input configurations 
$\v \in [-1, 1]^{d \cdot n}$ and 
$\v' \in [-1, 1]^{d \cdot n}$.

\begin{definition}[Local differential privacy (LDP)]
    \label{def:dp}
	A non-interactive mechanism $\algM$ satisfies $(\eps,\delta)$-LDP
w.r.t. the neighboring relation $\sim$, 
iff for any two input configurations 
$\v \in [-1, 1]^{d \cdot n}$ and 
$\v' \in [-1, 1]^{d \cdot n}$ 
such that $\v \sim \v'$, 
it holds that 
	$$\Pr[\view_{\algM}(\v) \in S] \le e^\eps \Pr[\view_{\algM}(\v') \in S]+\delta$$ 
where $\view_{\algM}(\v)$ is a random variable
representing the server's view upon input configuration $\v$; in particular,
the view consists of all messages received by the server.
\end{definition}
If a mechanism satisfies $(\epsilon, 0)$-LDP, 
we also say 
that it satisfies $\epsilon$-LDP
(w.r.t. to some neighboring relation $\sim$).

\begin{definition}[Event-level LDP]
    \label{def:item-ldp}
We say that a mechanism satisfies $(\epsilon, \delta)$-event-level-LDP, 
iff it 
satisfies $(\epsilon, \delta)$-LDP w.r.t. the following neighboring relationship: 
two input configurations $\v = (v_1, \ldots, v_n) \in [-1, 1]^{d \cdot n}$
and $\v' = (v'_1, \ldots, v'_n) \in [-1, 1]^{d \cdot n}$
are considered neighboring, 
iff they differ in at most one position (i.e., one coordinate contributed
by one user).
\ignore{
iff the two vectors are otherwise identical except for  
at most one user's contribution $v_i$ and $v'_i$;
and further,  for the user $i$ where the two vectors differ, it must be
that $\Vert v_i - v'_i \Vert_0 \leq 1$.
}
\ignore{
    Define the neighboring relation as following: for two $k$ sparse vectors $v, v' \in [-1,1]^d$, $v\sim v'$ iff they only differ in one coordinate, i.e., 
    \begin{align*}
        \Vert v - v' \Vert_0 \le 1
    \end{align*}
    A non-interactive framework $\algM$ is $(\eps,\delta)$ \textbf{event-level local differentially private} if it is $(\eps,\delta)$-differential private w.r.t. to the neighboring relation defined above.  
}
\end{definition}

\begin{definition}[User-level LDP]
    \label{def:user-ldp}
We say that a mechanism satisfies $(\epsilon, \delta)$-user-level-LDP, 
iff it 
satisfies $(\epsilon, \delta)$-LDP w.r.t. the following neighboring relationship: 
two 
input configurations $\v$ and $\v'$ are considered neighboring
if they differ in at most one user's contribution.
\ignore{
    Define the neighboring relation as following: for any two $k$ sparse vectors $v, v' \in [-1,1]^d$, $v\sim v'$.
    
    A non-interactive framework $\algM$ is $(\eps,\delta)$ \textbf{user-level local differentially private} if it is $(\eps,\delta)$-differential private w.r.t. to the neighboring relation defined above.     
}
\end{definition}

\begin{definition}[LDP for $L$-neighboring]
    \label{def:user-ldp}
We say that a mechanism satisfies 
$(\epsilon, \delta)$-LDP
for $L$-neighboring, 
iff 
it satisfies $(\epsilon, \delta)$-LDP w.r.t. the following $L$-neighboring notion: 
two 
input configurations $\v$ and $\v'$ are considered $L$-neighboring, iff 
the two vectors are otherwise identical except for  
at most one user's contribution $v_i$ and $v'_i$;
and further,  for the user $i$ where the two vectors differ, it must be
that $\Vert v_i - v'_i \Vert_1 \leq L$.
\ignore{
    Define the neighboring relation as following: for two $k$ sparse vectors $v, v' \in [-1,1]^d$, $v\sim v'$ iff 
    \begin{align*}
        \Vert v - v' \Vert_1 \le L 
    \end{align*}
    
    A non-interactive framework $\algM$ is $(\eps,\delta)$ $L$-locally differentially private if it is $(\eps,\delta)$-differential private w.r.t. to the neighboring relation defined above.   }  
\end{definition}

For the case of $k$-sparse binary vectors where each client's $v_i \in \{0, 1\}^d$, 
the following simple facts hold.
A mechanism satisfies $(\epsilon, \delta)$-LDP for $1$-neighboring, 
if and only if it is $(\epsilon, \delta)$-event-level-LDP.
A mechanism satisfies $(\epsilon, \delta)$-LDP
for $k$-neighboring, if and only if it satisfies
$(\epsilon, \delta)$-user-level-LDP.
More generally, 
for the case of $k$-sparse real-valued vectors where each client's $v_i \in [-1, 1]^d$, 
the following facts hold.
If a mechanism satisfies $(\epsilon, \delta)$-LDP
for $2$-neighboring, 
it must also satisfy 
$(\epsilon, \delta)$-event-level-LDP.
If a mechanism satisfies 
$(\epsilon, \delta)$-LDP
for $2k$-neighboring,  
it must also satisfy  
$(\epsilon, \delta)$-user-level-LDP.

\ignore{
The $L$-LDP definition captures the essence of event-level LDP and user-level LDP by setting $L=2$ and $L=2k$ correspondingly. 

Our goal is to design a non-interactive LDP framework $\algM$ c that consists of a local randomizer $\Q$ and an aggregation algorithm $\cA$, such that $\cA$ outputs an estimation $\hat{v}$ for the mean vector $\bar{v}$. The framework should work for event-level LDP, user-level LDP and the $L$-LDP. We wish to find an $\algM$ such that the error $\Vert \hat{v} - \bar{v}\Vert$ is minimized and the communication cost between the clients and the server is low. We will mainly focus on the $L_\infty$ error throughout the paper. 
}

Throughout the paper, unless otherwise noted,
we use $L_\infty$-error to 
characterize the 
utility of our vector mean estimation mechanism.  
Specifically, 
$L_\infty$-error takes the maximum absolute error  
over all 
$d$ coordinates.

For the special case where $k=1$, 
Bassily and Smith~\cite{bassily2015local} proved the following lower bound
on the error of any $(\epsilon, \delta)$-event-level-LDP mechanism ---
note also that for the case $k = 1$, 
event-level and user-level LDP are the same up to a constant factor.

\begin{theorem}[Lower bound on the error of single-item frequency 
estimation~\cite{bassily2015local}]
	\label{theorem:one_item_lower_bound}
Suppose that $k = 1$.
For any $\eps=O(1)$ and $0\le \delta \le o(\frac{\eps}{n\log n})$,
any non-interactive mechanism 
that satisfies $(\epsilon, \delta)$-event-level-LDP
must incur 
expected $L_\infty$ error of magnitude at least  
	$$\Omega\left(\min \left(\sqrt{\frac{\log(d)}{\eps^2n}}, 1\right)\right).$$ 

\ignore{
For any one round algorithm $\algM=(\Q,\cA)$, such that the local randomizer $\Q$ is $\eps$-DP, there exists a distribution $P\in\mathrm{simplex}(d)$ from which the inputs $v_i,i\in[n]$ are sampled i.i.d such that the expected $L_\infty$ error of $\cA$ w.r.t $P$ is 
	$$\Omega\left(\min \left(\sqrt{\frac{\log(d)}{\eps^2n}}, 1\right)\right).$$ 
}
\end{theorem}

\ignore{
\begin{theorem}[Lower Bound on the Error of the Single Item Frequency Estimation]\cite{bassily2015local}
	\label{theorem:one_item_lower_bound}
	For any $\eps=O(1)$ and $0\le \delta \le o(\frac{1}{n\log n})$. For any one round algorithm $\algM=(\Q,\cA)$, such that the local randomizer $\Q$ is $\eps$-DP, there exists a distribution $P\in\mathrm{simplex}(d)$ from which the inputs $v_i,i\in[n]$ are sampled i.i.d such that the expected $L_\infty$ error of $\cA$ w.r.t $P$ is 
	$$\Omega\left(\min \left(\sqrt{\frac{\log(d)}{\eps^2n}}, 1\right)\right).$$ 
\end{theorem}
}

%% file: tex/sparse_vec.tex
\section{Sparse Vector Mean Estimation}

\subsection{Algorithm}
We give a unified 
algorithm that can be parametrized
to achieve either event-level or user-level
LDP, or LDP under $L$-neighboring.
Our proposed algorithm is presented in Algorithm~\ref{alg:sparse-vector-agg}. 

\begin{table*}[t]
    \centering
    \label{tab:error}
    \begin{adjustbox}{width=\linewidth}
    \begin{tabular}{c|ccccc}
\toprule
\textbf{Cases}     & \textbf{\#Bins $b$}  & \textbf{Clipping Range $\eta$} & \textbf{Laplacian Noise Magnitude $\Delta$}       & \textbf{$L_\infty$ Error}                                                                    & \textbf{Comm. Cost}           \\ \midrule
Event-level LDP($L=2$)    & $\frac{\eps^2 k}{4}$                  & $\infty$                              & $2$                         & $O(\frac{1}{\eps}\sqrt{\frac{\log(d/\beta)}{n}})$                                 & $O(k\log k)$                  \\
$L\le \sqrt[3]{k}$ & $\frac{\eps^2 k}{L^2}$      & $\infty$                              & $L$                         & $O(\frac{1}{\eps}L\sqrt{\frac{\log(d/\beta)}{n}})$                                & $O(\frac{k\log k}{L^2})$      \\
\\[-0.5em]
$L\ge \sqrt[3]{k}$ & $\sqrt{\frac{\eps^2 k}{L\log \frac{1}{\delta}}}$ & $\infty$                              & $3\sqrt{bL\log(2b/\delta)}$ & $O(\frac{1}{\eps}(kL\log(k/L\delta))^{\frac{1}{4}}\sqrt{\frac{\log(d/\beta)}{n}})$ & $O(\sqrt{\frac{k}{L}}\log k)$ \\
User-level LDP ($L=2k$)    & 1                    & $\sqrt{2k\log (4n/\beta)}$    & $2\eta$                     & $O(\frac{1}{\eps}\sqrt{k\log (n/\beta)}\sqrt{\frac{\log(d/\beta)}{n}})$                   & $O(\log k)$                   \\ \bottomrule
\end{tabular}
\end{adjustbox}
     \vspace{1ex}
     
    \caption{Summary of the parameters under different LDP settings, the corresponding utility guarantee and the communication costs.}
\end{table*}

\label{sec:sparse-vec}

\begin{algorithm}
	\caption{Non-interactive Algorithm for $k$-Sparse Vector Mean Estimation}\label{alg:sparse-vector-agg}
    Parametrize bin number $b$, 
    clipping range $\eta$,
    noise parameter $\Delta$ according to Table~\ref{tab:error}\;
	\textbf{Client-side algorithm given input vector $v_i$: }\\
		Randomly pick a hash function $h_i:[d]\to [b]$ \\
	Randomly pick a hash function $s_i:[d]\to \{-1,+1\}$ \\ 
	\For{$j \in [b]$}{
		$B_{i,j}\get \sum_{l\in[d],h_i(l)=i} s_i(l)v_l$ \\
        $\bar{B}_{i,j}=\textit{clip}_{[-\eta,+\eta]}(B_{i,j})$ \\
{\it /*clipping needed only for pure user-level LDP */} \\
        $\tilde{B}_{i,j} \get \bar{B}_{i,j} + {\sf Lap}(\frac{\Delta}{\eps})$ 
	}
	Send $(h_i,s_i,\tB_{i, 1},\dots,\tB_{i, b})$ to the server\\
	\textbf{Server-side algorithm:} \\
	For all coordinate $x\in[d]$:
	$\hat{v}_x \get \frac{1}{n}\sum 
    _{i\in[n]}s_{i}(x) \tB_{i, h_i(x)}$ \\
\end{algorithm}

In the above algorithm, the clipping
algorithm is needed only if we want to 
achieve $(\epsilon, 0)$-DP under user-level LDP --- see
Section~\ref{sec:improving} for more details.
For all other cases, we achieve $(\epsilon, \delta)$-LDP.

Further, in the above
algorithm, we assumed that the server computes the
entire mean vector. 
However, when the universe size $d$ is very large (e.g., the space
of all possible URLs), the server may not want to write
down the entire mean vector. Instead, it may wish
to query $\hat{v}_x$ for a specific item $x \in [d]$, 
e.g., the frequency of a specific URL.  
In this case, the server need not iterate through
every $x \in [d]$, it only needs to invoke 
Line 13
for the items $x\in [d]$
that it cares about.

\ignore{
\begin{remark}
The server-side algorithm compute the mean estimation for each coordinate individually. When $d$ is very large, the server usually wants to compute the mean values for a small subset of $X\subset [d]$. Then, the server only need to compute $\hat{v}_x$ for $x\in X$ correspondingly.
\end{remark}
}

\paragraph{Discreting real-numbers 
for transmission.}
In the above algorithm, we assumed that the client
is transmitting real-valued numbers to the server.
In practice, we can truncate and discretize real-valued
numbers before transmitting, 
and ensure that 
the per-client communication cost is only $O(b\log k)$.
The additional error introduced in the discretization process
is asymptotically absorbed by the existing error terms, 
and therefore this step does not introduce any additional asymptotical error.
See Appendix~\ref{sec:communication}
for details.


\ignore{
\begin{remark}
In order to efficiently transmit hash functions and real values, a pre-communication processing is needed. 
It is omitted in Algorithm~\ref{alg:sparse-vector-agg} and the detailed description is deferred to Appendix~\ref{sec:communication}.
Given proper truncating and discretization, the process satisfies the privacy requirement and maintains the same utility guarantee asymptotically. The per-client communication cost is $O(b\log k)$ . 
\end{remark}
}

    

\begin{theorem}[Main theorem]
    \label{thm:main-utility}
	 Assuming $nk/b\ge \log(5d/\beta)$. Assuming the hash functions $h$ and $s$ are random oracles. Algorithm~\ref{alg:sparse-vector-agg} satisfies $L$-neighboring $(\eps,\delta)$-LDP. 
	 With probability at least $1-\beta$, the algorithm~\ref{alg:sparse-vector-agg} outputs an estimation $\hat{v}$ with $L_\infty$ error of $O\left(\left(\sqrt{\frac{k}{b} }+\frac{\Delta}{\eps}\right)\sqrt{\frac{\log(d/\beta)}{n}}\right)$. The per-client communication cost is $O(b\log k)$.
\end{theorem}
Note that in practice, we can instantiate $h$ and $s$ with pseudorandom functions
(PRFs)
rather than random oracles. As mentioned earlier, 
the computational assumption here 
is not needed for the privacy but only for measure concentration. 

We present the privacy-related proof in Section~\ref{sec:privacy-analysis} and the utility-related proof in Section~\ref{sec:utility-analysis}. For the communication cost analysis, see the discussion in Appendix~\ref{sec:communication}.

\subsection{Privacy Analysis}
\label{sec:privacy-analysis}

Notice that the general $L$-neighboring LDP notion captures the requirement of event-level LDP($L=2$) and user-level LDP($L=2k$).
Therefore, we only need to prove our algorithm is $L$-neighboring LDP and instantiate with corresponding $L$ value for event- and user-level LDP. For now, we take the bin number $b$ as an unspecified variable and we will provide the optimal selection of $b$ later in the utility section.

Given two neighboring input configuration $\v,\v'$, from which one client's inputs are different, denoted as vectors $v,v'$. 
Then, $\Vert v-v' \Vert_1 \le L$.
We wish to bound the $L_1$ difference for the ``raw 
bin values'' $B_1,\dots,B_b$ and $B'_1,\dots,B'_b$
generated by two independent invocations of the client's algorithm. 

\begin{claim}
\label{claim:lap-noise}
Given any two neighboring vectors $v,v'$. 
Taking the randomness of $h$ and $s$,
if $\Pr[\sum_{j\in[b]}|B_j - B'_j| \ge \Delta]\le \delta$,
then adding Laplacian noise of $\Lap(\frac{\Delta}{\eps})$ ensures the two invocations of the client-side algorithm's output
distributions are $(\eps,\delta)$-close.
\end{claim}

The proof is simple that one can compute the privacy budget loss in each bin and the total budget will be bounded by $\eps$. 
We defer the proof to the appendix~\ref{sec:privacy-proof}.

For any $h, s$, rewrite $\sum_{j\in[b]}|B_j-B'_j|=  \sum_{j\in[b]}\left|\sum_{l\in[d], h(l)=j} s(l)(v_j-v'_j)\right|$. 
With absolute inequality, the $s(l)$ term can be removed and the above expression is at most $\sum_{j\in[b]}\sum_{l\in[d], h(l)=j} |v_j-v'_j| $, which is exactly $L$. This proves the privacy property when $L\le k^{\frac{1}{3}}$ where we set $\Delta=L$.

Moreover, we want to further prove that the difference after the binning is bounded by $\tilde{O}(\sqrt{bL})$. 
The intuition is that the binning process ``squeezes'' the difference, so that we can add smaller noise.
For example, let's say a pair of neighboring vectors $v,v'$ differ in coordinates $x$ and $y$. Say $v_x=v_y=1$ and $v'_x=v'_y=-1$.
The original $L_1$ difference in the two coordinates are 4.
Suppose the client samples a hash function $h$ such that $h(x)=h(y)$.
Then, we know that with probability $1/2$, the random $\pm 1$ function $s$ turns out to have outputs that $s(x)=-s(y)$.
In this case, the influence in coordinates $x$ and $y$ , i.e., $|v_x-v'_x|$ and $|v_y-v_y'|$, cancel each other out in the bin because $|(s(x)v_x + s(y)v_y) - (s(x)v'_x + s(y)v'_y)|=0$. We formally claim the following lemma: 
\begin{lemma}
\label{claim:squeeze}
Assuming $L/b \ge \log(2b/\delta)$. 
Consider any two neighboring vectors $v,v'\in [-1,+1]^d$ such that $\Vert v-v'\Vert_1 \le L$.  
We have that $\Pr\left[\sum_{j\in[b]}|B_j-B'_j|> 3\sqrt{bL\log (2b/\delta)}\right]\le \delta$. 
\end{lemma}

\begin{proof}
    Fix a bin $j$. Define random variables $ Z_l=\I[h(l)=j] s(l)(v_l-v'_l)$ for $l\in[d]$ and $Z=\sum_{l\in [d]} Z_l$. $Z$'s distribution is exactly the difference in bin $j$ after bining. We know that the variables $\{Z_l\}_{l\in[d]}$ are independent and bounded by $[-2,2]$. Also, $\E[Z_l^2] = \frac{1}{b} (v_l-v'_l)^2 \le \frac{2}{b} (v_l - v'_l)$. Let $\mu=\frac{L}{b}$. Using Bernstein's inequality, setting $t=3\sqrt{\mu\log (2b/\delta)}$,  we have

    \begin{align*}
        & \Pr[|Z| \ge t] 
        \le  2\exp\left(-\frac{\frac{1}{2}t^2}{\sum_{l\in [d]}\E[Z_l^2] + \frac{2}{3}t}\right) \\
        \le & 2\exp\left(-\frac{\frac{9}{2}\mu\log (2b/\delta)}{2\mu + 2\sqrt{\mu\log (2b/\delta)}}\right) 
        =  2(\delta/2b)^{\frac{\frac{9}{2}\mu}{2\mu + 2\sqrt{\mu\log (2b/\delta)}}}
    \end{align*}
    
    Using the condition that $\mu=\frac{L}{b}\ge \log(b/\delta)$, we have $\frac{\frac{9}{2}\mu}{2\mu + 2\sqrt{\mu\log (2b/\delta)}} \ge \frac{9}{8} \ge 1$. Therefore, $\Pr\left[|Z|\ge 3\sqrt{\mu \log (2b/\delta)}\right]\le \delta/b$. Taking the union bound over all $b$ bins, we have the total difference in all bins are at most $3\sqrt{bL\log(2b/\delta)}$ with probability at least $1-\delta$.
\end{proof}

By Claim~\ref{claim:lap-noise} and Lemma~\ref{claim:squeeze}, setting the noise parameter to $\Delta= 3\sqrt{bL\log(2b/\delta)}$ is enough to achieve $(\eps,\delta)$-privacy under $L$-Neighboring LDP, assuming $L/b\ge \log(2b/\delta)$. 
In practice, when we search for the optimal $b$, we will carefully set $b$ such that the condition $L/b\ge \log(2b/\delta)$ is held.

\subsection{Utility Analysis}
\label{sec:utility-analysis}

We first provide the simplified version of the utility part for the main theorem for general parameter settings -- bin number $b$, clipping range $\eta$ and the noise parameter $\Delta$.
The full proof is deferred to Appendix~\ref{sec:full-proof}.
Then, we will discuss how to choose the optimal $b$ to achieve the best utility under different scenarios. 

\begin{proof}
\textit{(Sketch).}

\noindent Fix an index $x\in[d]$.
The server computes the estimation
$\hat{v}_x$ as $\hat{v}_x = \frac{1}{n}\sum_{i\in[n]} s_i(x)\tB_{i, h_i(x)}$. 
We bound the error by the three steps: 1) binning; 2) clipping; 3) adding Laplacian noise.

\paragraph{Binning error.} The absolute error of binning is $\left|\frac{1}{n}\sum_{i\in[n]} v_{i,x} - \frac{1}{n}\sum_{i\in[n]} B_{i, h_i(x)}s_i(x)\right|$. 
Define random variables $Y_{i,l}=\I[h_i(l)=h_i(x)]v_{i,l}s_i(l)s_i(x)$ as the error introduced by coordinate $l\ne x$ in client $i$'s vector. 
The error is equal to $|\frac{1}{n}\sum_{i\in[n]}\sum_{l\in[d],l\ne x}Y_{i,l}|$.
Since we model the hash function $h_i$ as random oracle, the hash collision probability is $\frac{1}{b}$. Also, the hash function $s_i$ is a uniform $\pm 1$ function, so we have $Y_{i,l}=v_{i,l}$ with prob. $\frac{1}{2b}$ and $Y_{i,l}=-v_{i,l}$ with prob. $\frac{1}{2b}$. 
We consider all $nk$ non-zero coordinates and we can use the analysis for a zero-mean random walk with total length of $\frac{nk}{b}$. 
Using Berstein's inequality, we can prove that for all coordinate $x$, $ |\frac{1}{n}\sum_{i\in[n]} v_{i,x} - \frac{1}{n}\sum_{i\in[n]} B_{i, h_i(x)}s_i(x)| =O(\sqrt{\frac{k}{b}}\sqrt{\frac{\log(d/\beta)}{n}})$ with probability at least $1-O(\beta)$.

\paragraph{Clipping error.} We actually prove that the clipping range is large enough, so that the clipping error is zero with high probability. We directly compute the raw bin value $B_{i,j}$'s moment generating function and conclude that it is a sub-Gaussian r.v. with a variance at most $k$.
That means the absolute bin values will roughly be $\tilde{O}(\sqrt{k})$. Set the clipping range to $\eta\ge\sqrt{2k\log(4nb/\beta)}$. Using concentration bound for sub-Gaussian variables and taking union bound over all $nb$ bins across $n$ clients, we conclude the probability of the clipping error being zero is at least $1-O(\beta)$.

\paragraph{Laplacian noise error.} 
Finally, we look at the absolute error term introduced by adding Laplacian noise. It turns out that the error's distribution is the same as the distribution for the mean of $n$ i.i.d. Laplacian variables with parameter $\frac{\Delta}{\eps}$. Using the concentration bound for Laplacian noise and taking the union bound over all $x\in[d]$, the maximal error is bounded by $O\left(\frac{\Delta}{\eps}\sqrt{\frac{\log(d/\beta)}{n}}\right)$ with probability $1-O(\beta)$.

Combine the above arguments. By taking the union bound and setting the constants appropriately, we can conclude that the $L_\infty$ error is $O\left(\left(\sqrt{\frac{k}{b} }+\frac{\Delta}{\eps}\right)\sqrt{\frac{\log(d/\beta)}{n}}\right)$ with probability $1-\beta$.
\end{proof}

Fix $k, L$. 
From the privacy analysis section, we know that $\Delta$ can be set to $\min(L, 3\sqrt{bL\log(2b/\delta)})$.
We now try to find the optimal $b$ to minimize the error. 
Define functions $f_1(b)=\sqrt{\frac{k}{b}} + \frac{L}{\eps}$ and $f_2(b)= \sqrt{\frac{k}{b}} + \sqrt{bL\log(2b/\delta)}/\eps$. 
The error can be rewritten as $\min(f_1(b) + f_2(b))\cdot \tilde{O}(\frac{1}{\sqrt{n}})$.
Optimizing $f_1(b)$, the optimal $b$ is $b^*_1=\frac{\eps^2 k}{L^2}$ and $f_1(b^*_1)=\frac{L}{\eps}$.
Optimizing $f_2(b)$, the optimal $b$ is roughly $b^*_2 = \sqrt{\frac{\eps^2k}{L\log \frac{1}{\delta}}}$ and $f_2(b^*_2)=O(\frac{1}{\eps} (kL\log(\frac{k}{L\delta}))^{\frac{1}{4}})$.
Then, by comparing two local minimums, we conclude that when $L\le k^{\frac{1}{3}}$, the optimal $b$ is $\frac{\eps^2 k}{L^2}$ and the error is $O\left(\frac{L}{\eps}\sqrt{\frac{\log (d/\beta)}{n}}\right)$. 
When $L\ge k^{\frac{1}{3}}$, the optimal $b$ is roughly $\sqrt{\frac{\eps^2k}{L\log \frac{1}{\delta}}}$ and the error is $O\left(\frac{1}{\eps}(kL\log(\frac{k}{L\delta}))^{\frac{1}{4}})\sqrt{\frac{\log (d/\beta)}{n}}\right)$. In practice, we also take concrete constants into consideration and select the best $b$ accordingly.

\subsection{Achieving $(\eps,0)$-user-level LDP}

\label{sec:improving}

We analyze how the extra clipping step in Algorithm~\ref{alg:sparse-vector-agg} achieves pure-LDP in user-level setting.
The idea is to push the ``failure probability'' $\delta$ in privacy definition to the utility theorem's failure probability $\beta$.
We first observe that in user-level LDP, we have the neighboring distance $L=2k$ and the optimal bin number selection is $b=1$. 
We now consider the magnitude of the ``raw bucket value'' $B_{i,1}$ for client each $i$. 
We simply have $B_{i,1}=\sum_{l \in [d]} v_{i,l}s_i(l)$.
Take the randomness of the random $\pm 1$ function $s_i$.
We can see the distribution of $B_{i,1}$ is similar to a zero-mean random walking with at most $k$ steps, where each step's length is at most 1. 
Using Berstein-type concentration bound, we can prove that with probability $1-O(\beta)$, for all client $i\in[n]$, $|B_{i,1}| \le \sqrt{2k\log(4n/\beta)}$. 
See the detail proof in Appendix~\ref{sec:full-proof}.
Let $\eta=\sqrt{2k\log(4n/\beta)}$. 
We see that the difference in $B_{i,1}$ and $B'_{i,1}$ in two independent invocation of the client-side algorithm given input $v, v'$ are at most $2\eta$ with probability 1 after clipping. 
That means we only need to set the noise parameter $\Delta=2\eta$ and the algorithm is $(\eps,0)$-user-level-LDP. 
Plug the parameters into the main theorem, we know the utility guarantee of this optimization is $O(\frac{1}{\eps}\sqrt{k\log (n/\beta)}\sqrt{\frac{\log(d/\beta)}{n}})$. 
We see the utility guarantee of this optimization is similar to the original unclipped version -- they are both $\tilde{O}(\frac{1}{\eps}\frac{\sqrt{k}}{\sqrt{n}})$. In practice, this clipped version has much smaller constant factor in terms of error.

%% file: tex/eval.tex
\section{Evaluation}
\label{sec:eval}


\subsection{Setup}

\paragraph{Implementation.}
To evaluate our approach, we implement it with C++, compile it with gcc4.8 and the C++11 standard. 
We use 40-bit random seeds to generate the hash functions. 
For simplicity, we directly use 32-bit floating numbers to store and transmit real values. 


\paragraph{Datasets.} We evaluate the algorithms for both synthetic and real-world datasets. For the synthetic dataset, we assume there are $10^5$ users, each with a vector of dimension $d$ and sparsity $k$. We first randomly sample the non-zero coordinates according to Zipf's distribution with a suitable degrading parameter ($s=1.4$). We choose the Zipf's distribution because it naturally appears in real-world data analytics. For each sampled non-zero coordinate, the actual value is sampled from a Gaussian distribution with mean $\mu=1$ and standard deviation $\sigma=0.3$. Then the values are clipped to $[-1,1]$. 
\begin{table}[t]
\centering
\begin{tabular}{l|c|c|c|c}
\toprule
\textbf{Datasets} & \textbf{\#Clients $n$} & \textbf{\#Items $d$} & \textbf{\#Records} & \textbf{Sparsity $k$} \\ \midrule
Clothing\cite{clothing}          & 47958                  & 1378                 & 79285              & 6                     \\ 
\rowcolor{mygray} Renting\cite{renting}           & 105571                 & 5850                 & 183052             & 11                    \\ 
Movies\cite{movie}           & 138493                 & 26744                 & 7019990            & 100                    \\ \bottomrule
\end{tabular}
\hfill

\label{tab:dataset}
\vspace{1ex}
\caption{Real-world dataset.}
\end{table}

For the real-world dataset experiment, we downloaded three open-sourced datasets from Kaggle, including an online cloth shopping dataset~\cite{clothing}, a clothing renting dataset~\cite{renting} and a movie rating dataset~\cite{movie}, where each record describes one activity (purchase, rent, or rating, respectively). Table~\ref{tab:dataset} gives more information about the datasets.  We select those records with client feedback ratings and normalize them to $[-1,1]$. Given the sparsity parameter is $k$, for clients with more than $k$ records, we randomly sample $k$ records. 

\paragraph{Metrics.} We consider both utility and communication cost fixing the privacy level (i.e., fixing $\eps$ and $\delta$). 
To measure utility, we use the $L_\infty$ error and the mean square error (MSE).  Given $\hat{v} = \frac{1}{n} \sum_{i\in[n]} v_i$ as the true mean vector and $\bar{v}$ as the estimation vector, they are defined as:
\begin{align*}
    L_\infty~\text{Error} &= \max_{x\in[d]} \left|\hat{v}_x - \bar{v}_x\right| \indent
    \text{MSE} &= \frac{1}{d}\sum_{x\in[d]} \left(\hat{v}_x - \bar{v}_x\right)^2 
\end{align*}
For the communication cost, we measure the per-client communication cost: We sum up the \textbf{byte-length} of all the reports from the clients and compute the average report size. 

\paragraph{Evaluation Roadmap.}
We split the experiments into three groups: user-level LDP setting, event-level setting LDP, and the $L$-Neighboring setting. Within each group, we measure different methods varying three parameters: sparsity $k$, privacy budget $\eps$ (in most cases, we use $\delta=0$; but when $\delta>0$, e.g., for the naive perturbation scheme with Gaussian noise, we always use $\delta=10^{-5}$), and dimension size $d$.
We mainly compare our proposed method with the $k$-fold repetition-plus-1-sparse mechanism (referred as \textit{$k$-fold repeition}), the sampling + 1-sparse mechanism (referred as \textit{sampling}), the naive pertubation mechanism (with Gaussian Noise~\cite{GaussianMech}), Harmony~\cite{harmony} and PCKV~\cite{pckv}. 
We run the experiment 10 times and report the average error and the average communication cost. 




\subsection{Performance under User-level LDP}

\begin{figure*}[t!]
\centering
    \subfigure[$L_\infty$ error results when fixing $n=10^5$, $d=4096$, $\eps=1.0$ and varying $k$ from 1 to 1024.\label{fig:vary-k-mae-user}]{
    \includegraphics[width=0.31\textwidth]{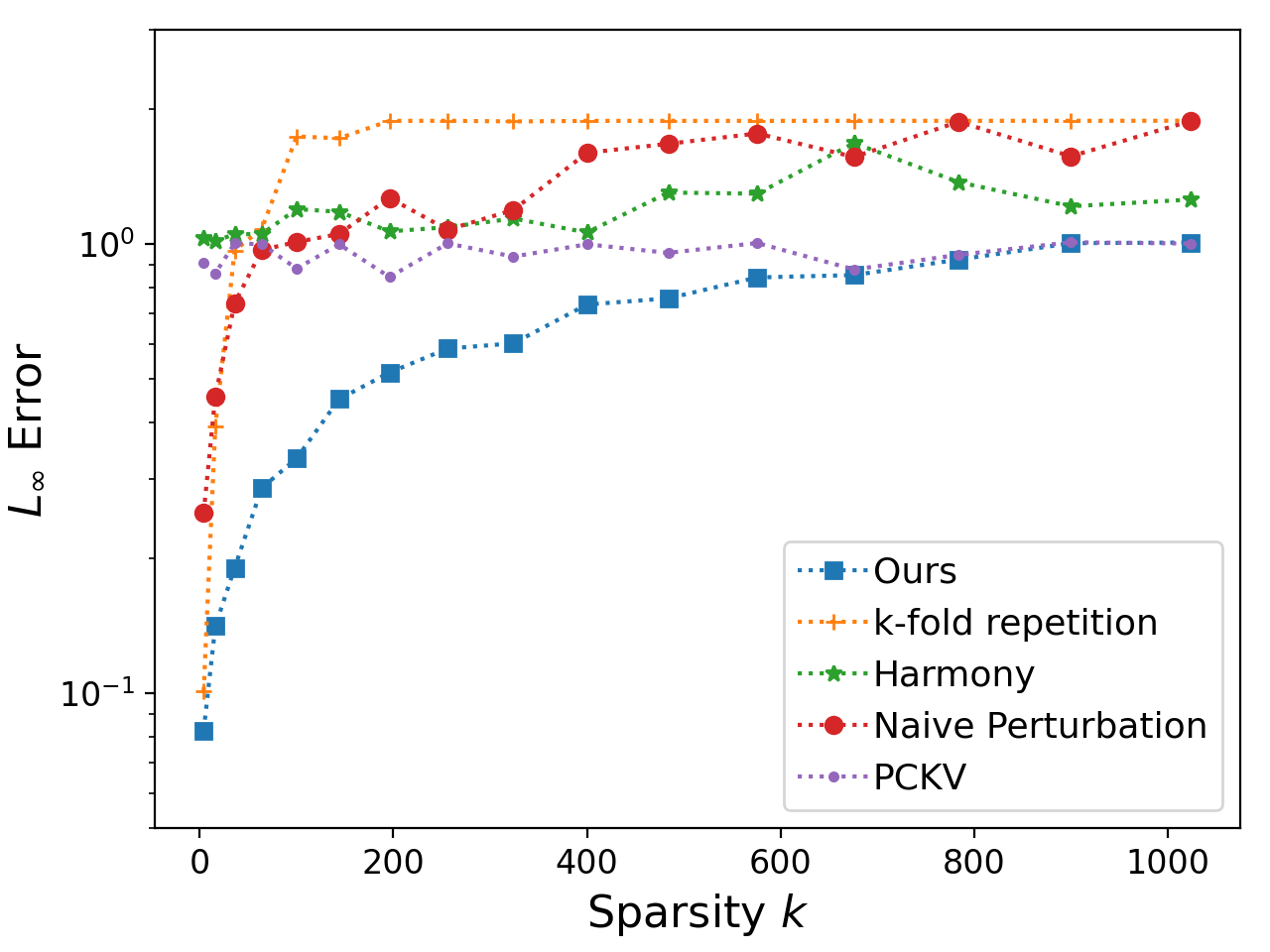}
    }
       \subfigure[$L_\infty$ error results for top 100 coordinates when fixing $n=10^5$, $d=10^5$, $k=64$ and varying $\eps$ from 0.5 to 3.5.\label{fig:vary-e-mae-user}]{
    \includegraphics[width=0.31\textwidth]{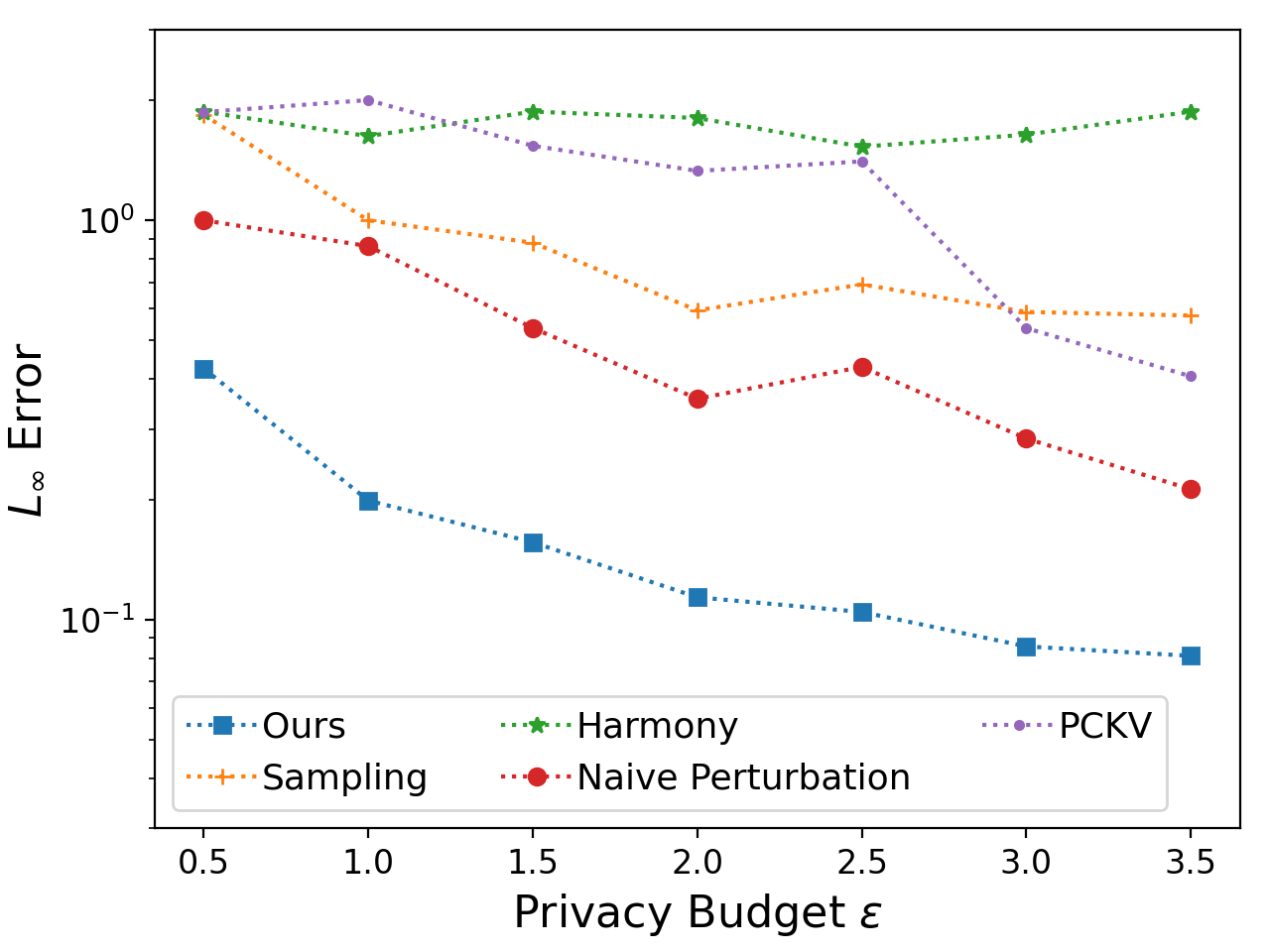}
    }
       \subfigure[$L_\infty$ error results for top 100 coordinates when fixing $n=10^5$, $k=64$, $\eps=1.0$ and varying $d$ from 64 to $10^5$.\label{fig:vary-d-mae-user}]{
    \includegraphics[width=0.31\textwidth]{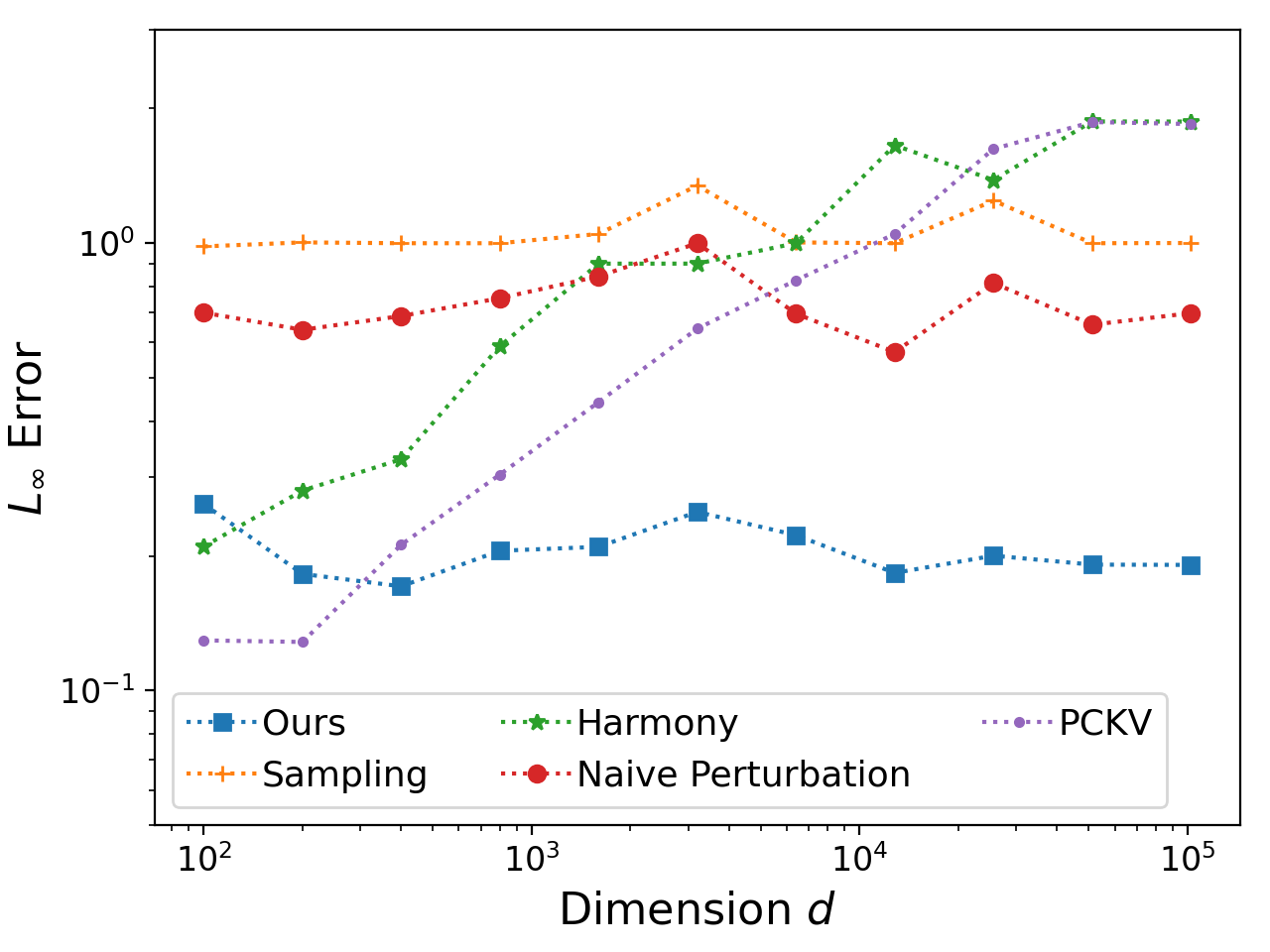}
    }\\
    \centering
       \subfigure[MSE results when fixing $n=10^5$, $d=4096$, $\eps=1.0$ and varying $k$ from 1 to 1024.\label{fig:vary-k-mse-user}]{
    \includegraphics[width=0.31\textwidth]{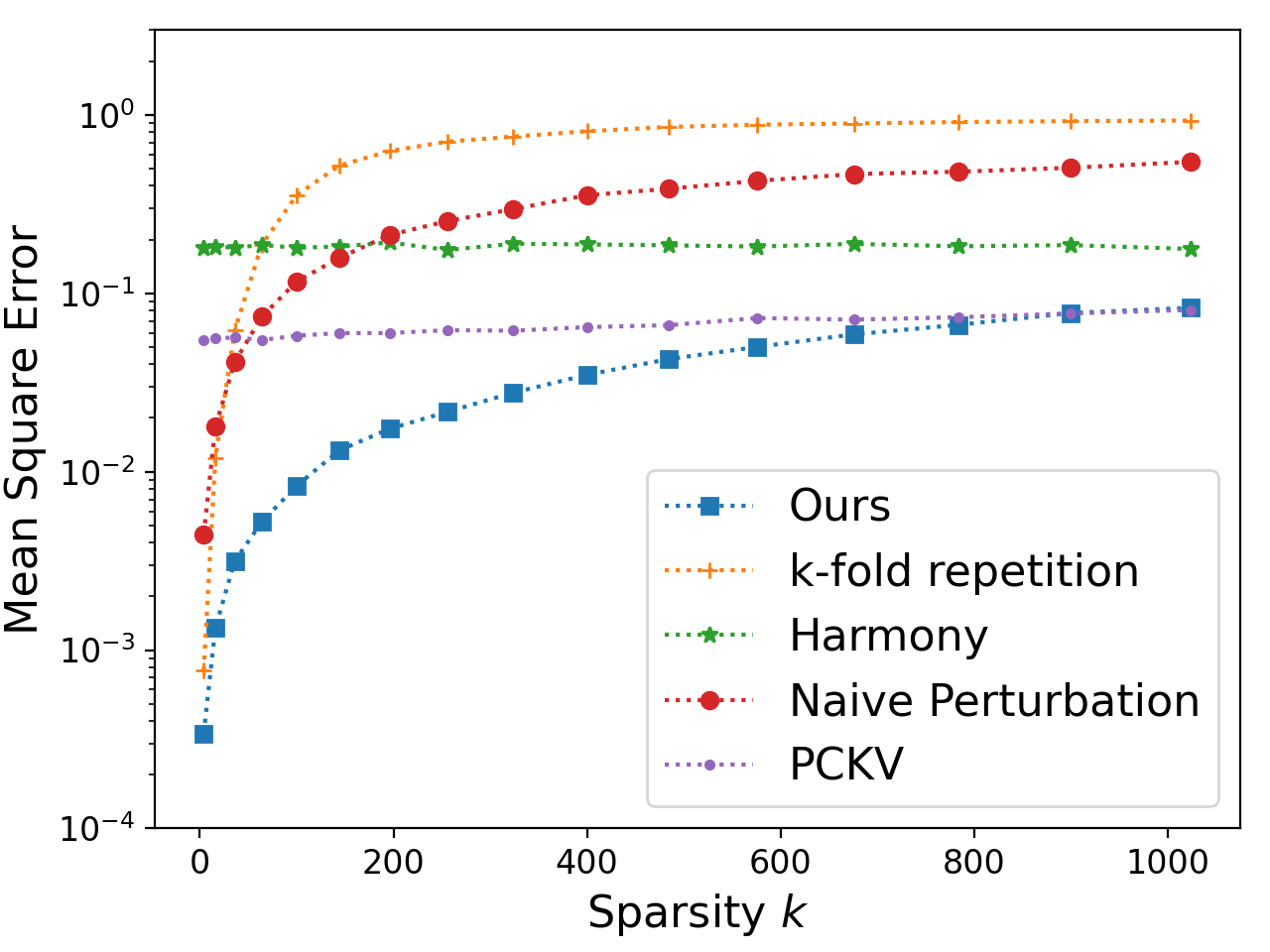}
    }
       \subfigure[MSE results results for top 100 coordinates  when fixing $n=10^5$, $d=10^5$, $k=64$ and varying $\eps$ from 0.5 to 3.5.\label{fig:vary-e-mse-user}]{
    \includegraphics[width=0.31\textwidth]{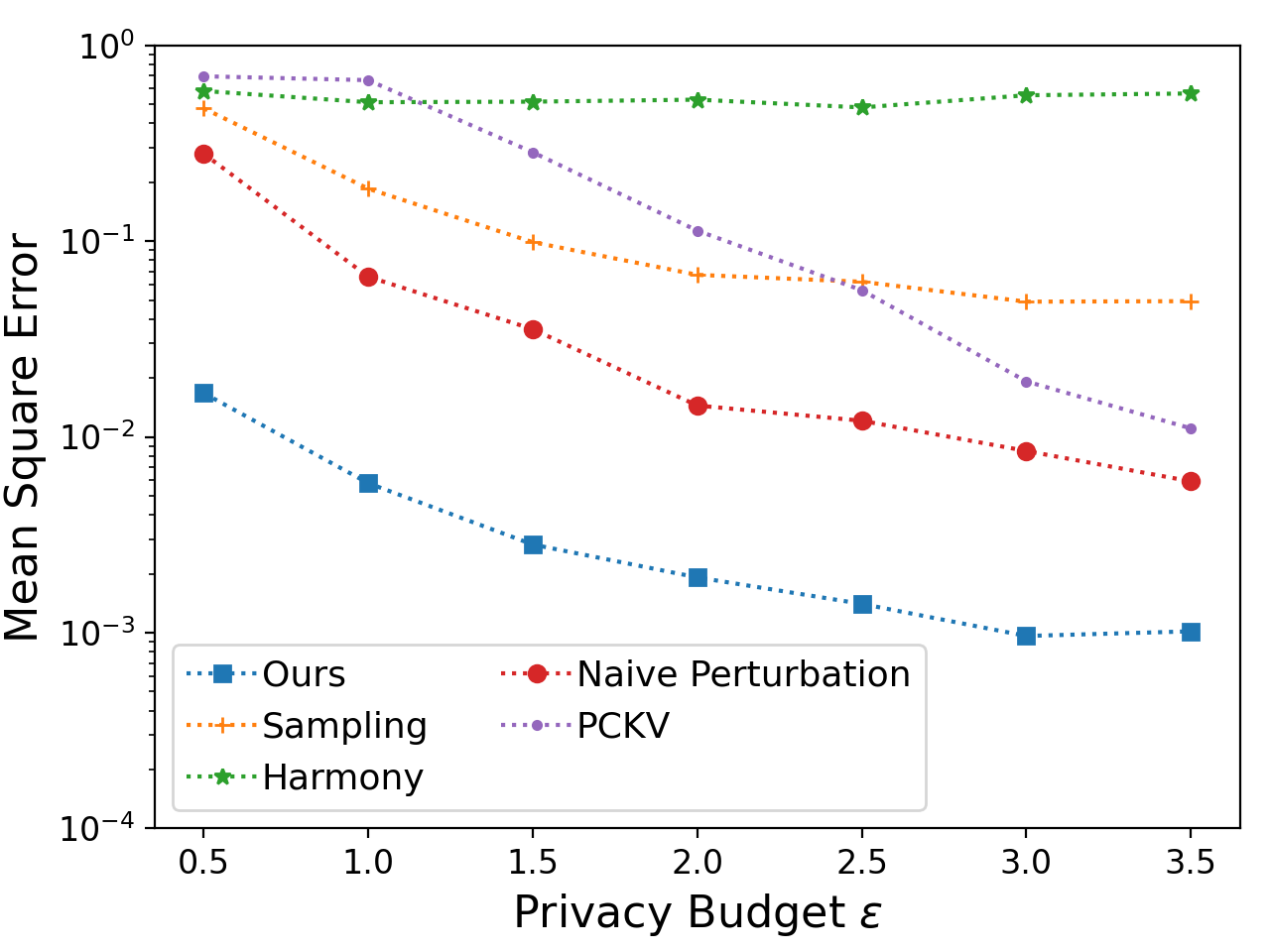}
    }
       \subfigure[MSE results for top 100 coordinates when fixing $n=10^5$, $k=64$, $\eps=1.0$ and varying $d$ from 64 to $10^5$.\label{fig:vary-d-mse-user}]{
    \includegraphics[width=0.31\textwidth]{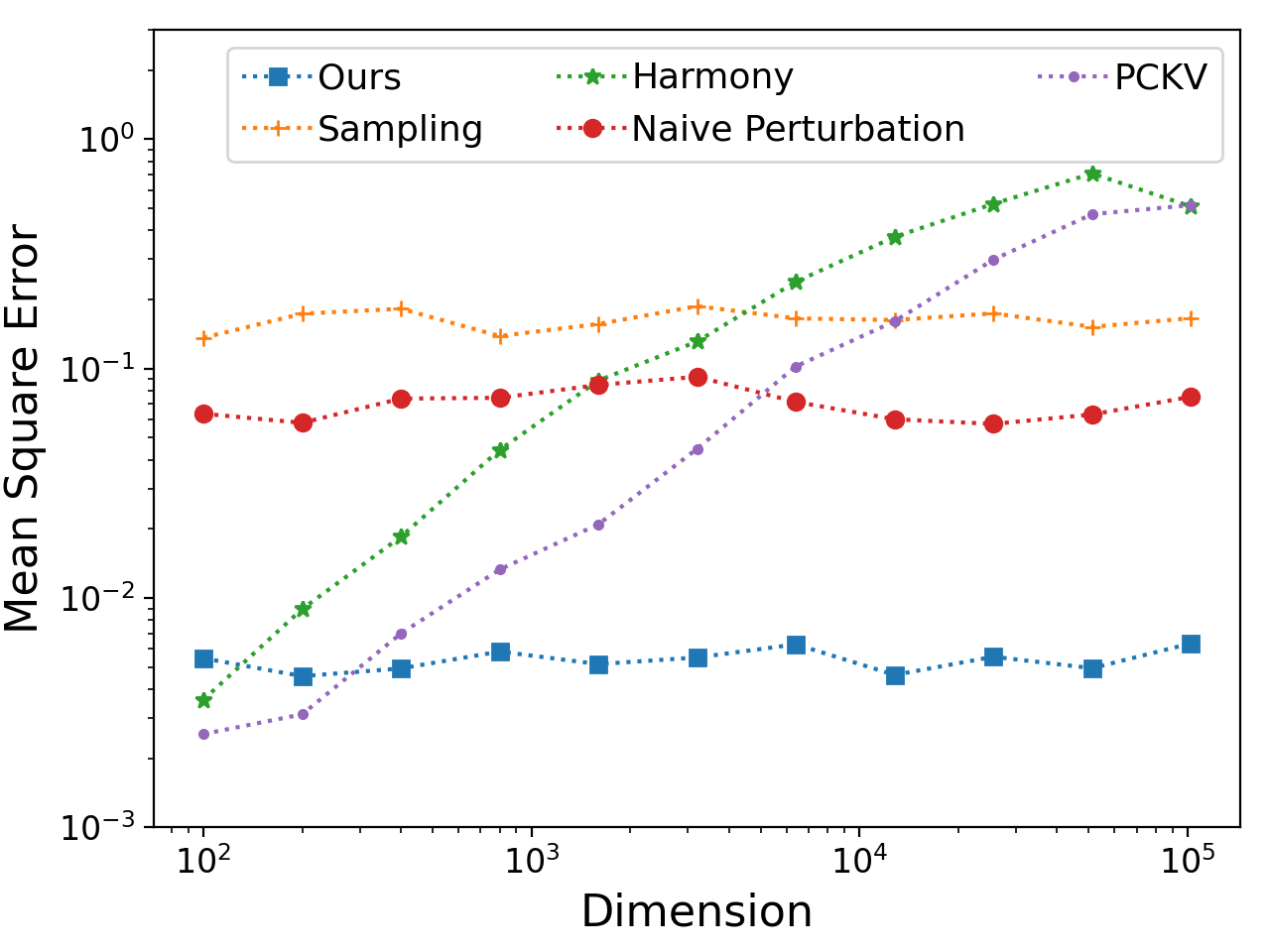}
    }
    \caption{Comparing the utilities of our method and existing approaches under user-level LDP. The subfigures in the top and bottom row show the results of $L_\infty$ error and MSE, respectively. The subfigures in the left, middle, and right column varies sparsity $k$ (from 1 to 1024), $\epsilon$ (from 0.5 to 3.5), and dimension $d$ (from 64 to $10^5$), respectively (while fixing the other two parameters).}
  \label{fig:main}
\end{figure*}

User-level LDP is the more standard setting in LDP analytics. Existing methods are mostly designed for user-level LDP.  We first compare our method against existing ones in this setting.

\paragraph{Varying sparsity $k$.} We plot the $L_\infty$ error results in Figure~\ref{fig:vary-k-mae-user} and the MSE results Figure~\ref{fig:vary-k-mse-user}. In our theoretical analysis, we prove that the $L_\infty$ error of our algorithm scales with $\sqrt{k}$. The sampling + 1-sparse method's error scales with $k$, and other algorithm cannot utilize the sparsity. The figures show that our method has the smallest estimation error for the whole region when $k$ ranges from 1 to 1024. The error of the sampling solution and the naive perturbation mechanism scales with $k$ and they perform worse than PCKV and Harmony when the sparsity $k$ is larger than $\sqrt{d}$. 

\paragraph{Varying privacy budget $\eps$.} The results are shown in Figure~\ref{fig:vary-e-mae-user} and Figure~\ref{fig:vary-e-mse-user}. With larger privacy budget, all schemes except Harmony achieve better estimation errors. However, when the dimension $d$ is sufficiently large, Harmony and PCKV suffer from a $\tilde{O}(\frac{\sqrt{d}}{\sqrt{n}})$ error. In the relatively high privacy budget region, PCKV shows better performance. Our method always has the smallest error in the reasonable large privacy budget range.

\paragraph{Varying dimension $d$.} 
In many use cases, the domain size (vector length) can be extremely huge, such as all possible products on Amazon, all possible URL and all geographical location on the earth. In this experiment, we only measure the top 100 coordinate with the largest absolute mean value. This is actually inspired by a real use case where the domain size is sufficiently and the server only wishes to compute the value for a limited keys (e.g. website access analysis). The results are shown in Figure~\ref{fig:vary-e-mae-user} and Figure~\ref{fig:vary-e-mse-user}. Our method provides an important feature -- its utility and communication cost decouple from the domain size. Our method can maintain a stable estimation error even with very large dimension $d$, while using minimum communication cost. The naive perturbation scheme needs to communicate $O(d)$ bits between the clients and the server. In the very dense case, where $k\approx d$, PCKV and Harmony has slightly better estimation error because our method has the extra $\sqrt{\log n}$ term in the error. However, in the more sparse case, all other methods fail to provide any meaningful guess. The noticeable drop in the large $d$ region of the error curves for PCKV and Harmony is because they basically output a meaningless zero vector. 

\subsection{Performance under Event-level LDP}

\begin{figure*}[t!]
\centering
    \subfigure[$L_\infty$ error results when fixing $n=10^5$, $d=4096$, $\eps=1.0$ and varying $k$ from 1 to 1024.\label{fig:vary-k-mae-item}]{
    \includegraphics[width=0.31\textwidth]{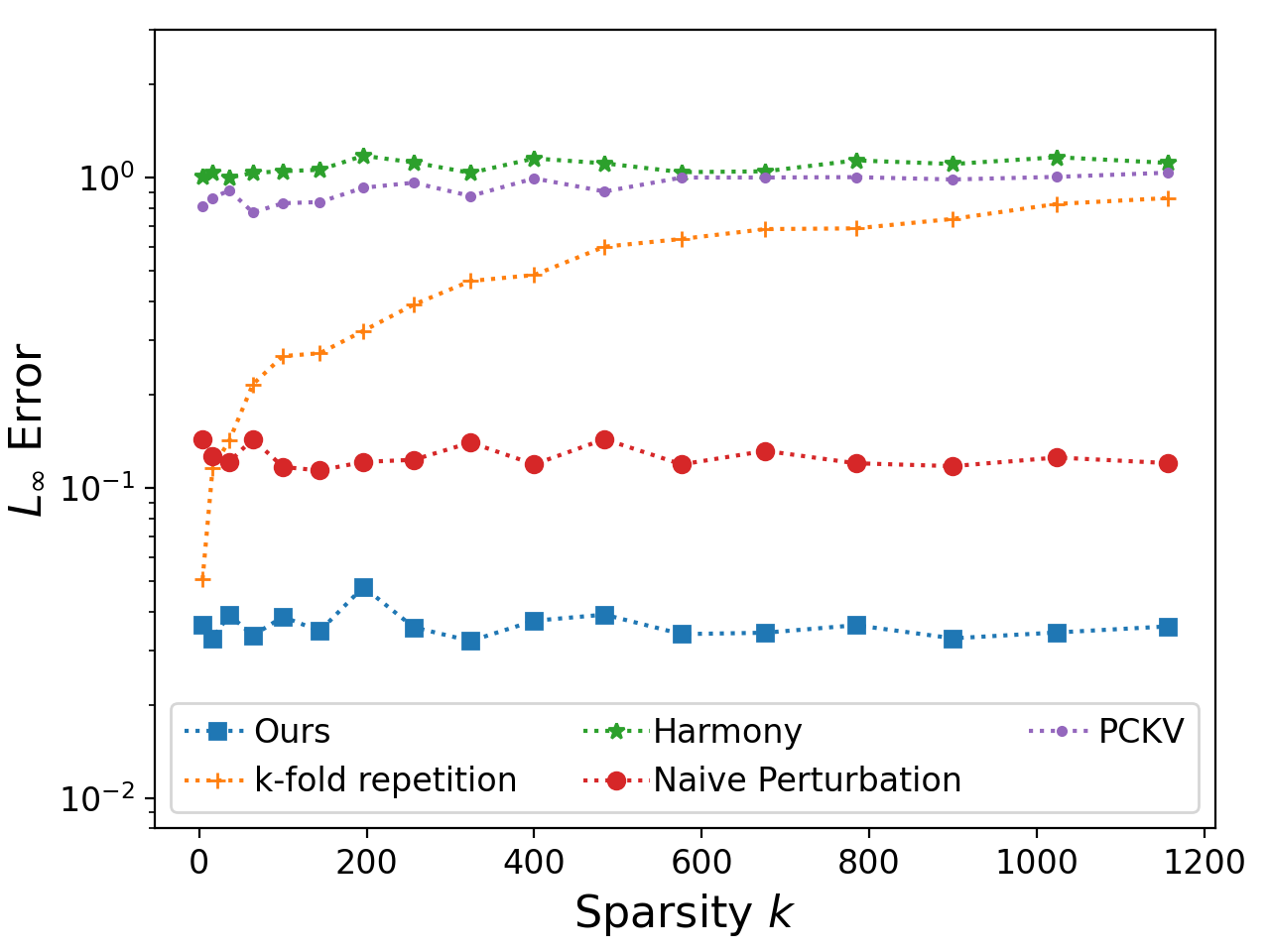}
    }
       \subfigure[$L_\infty$ error results for top 100 coordinates when fixing $n=10^5$, $d=10^5$, $k=64$ and varying $\eps$ from 0.5 to 3.5.\label{fig:vary-e-mae-item}]{
    \includegraphics[width=0.31\textwidth]{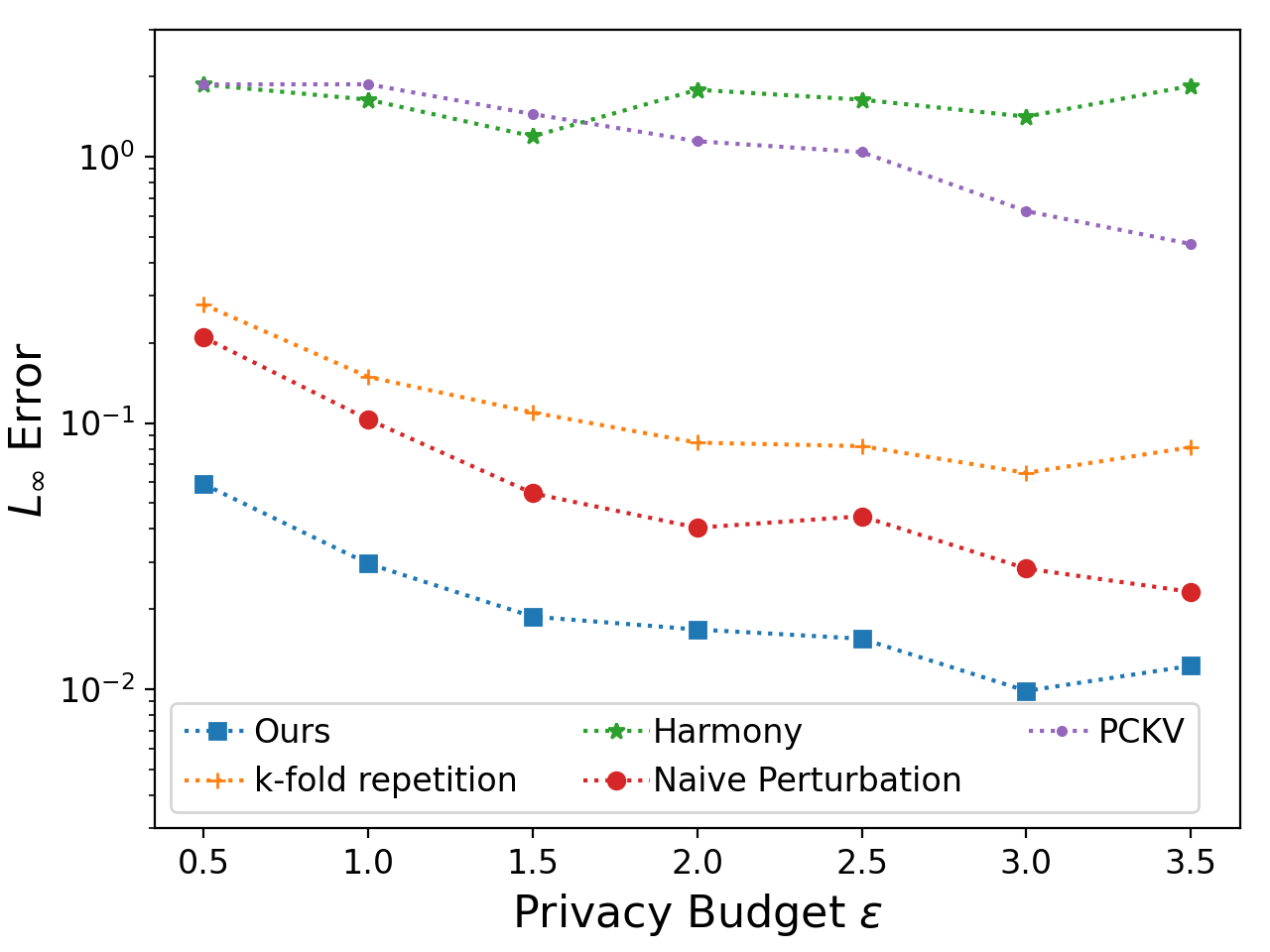}
    }
       \subfigure[$L_\infty$ error results for top 100 coordinates when fixing $n=10^5$, $k=64$, $\eps=1.0$ and varying $d$ from 64 to $10^5$.\label{fig:vary-d-mae-item}]{
    \includegraphics[width=0.31\textwidth]{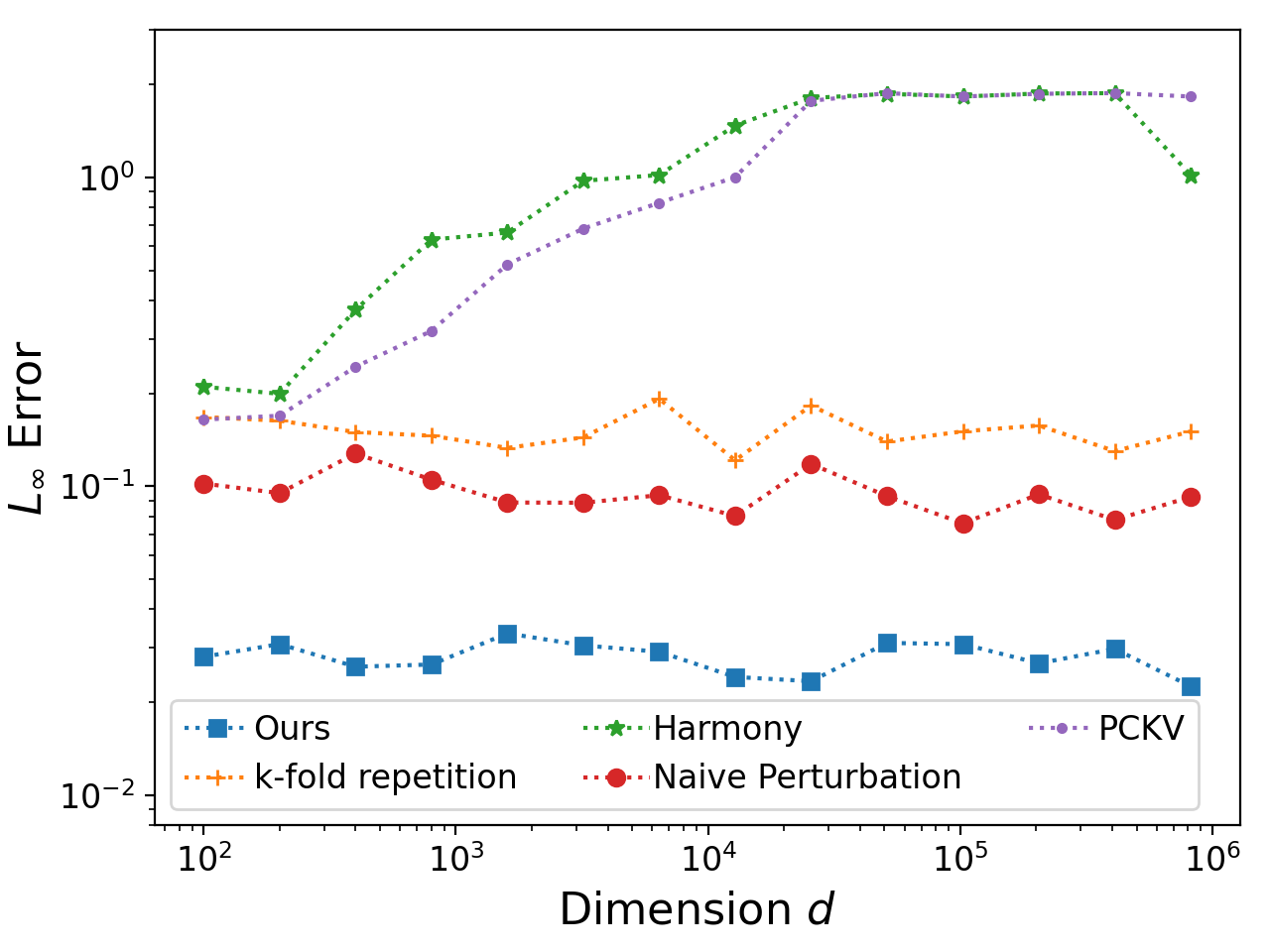}
    }
       \subfigure[MSE results when fixing $n=10^5$, $d=4096$, $\eps=1.0$ and varying $k$ from 1 to 1024.\label{fig:vary-k-mse-item}]{
    \includegraphics[width=0.31\textwidth]{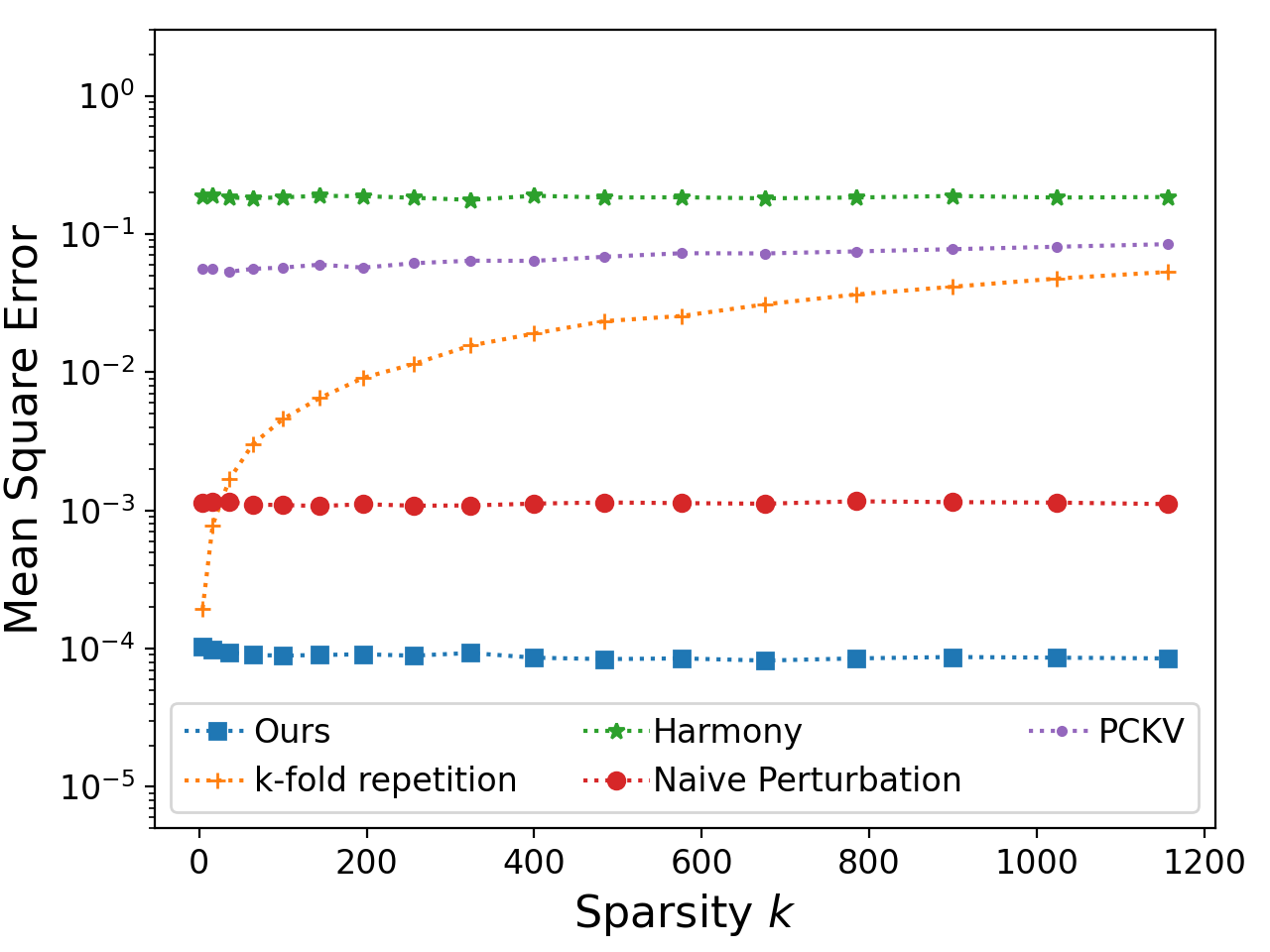}
    }
       \subfigure[MSE results for top 100 coordinates when fixing $n=10^5$, $d=10^5$, $k=64$ and varying $\eps$ from 0.5 to 3.5.\label{fig:vary-e-mse-item}]{
    \includegraphics[width=0.31\textwidth]{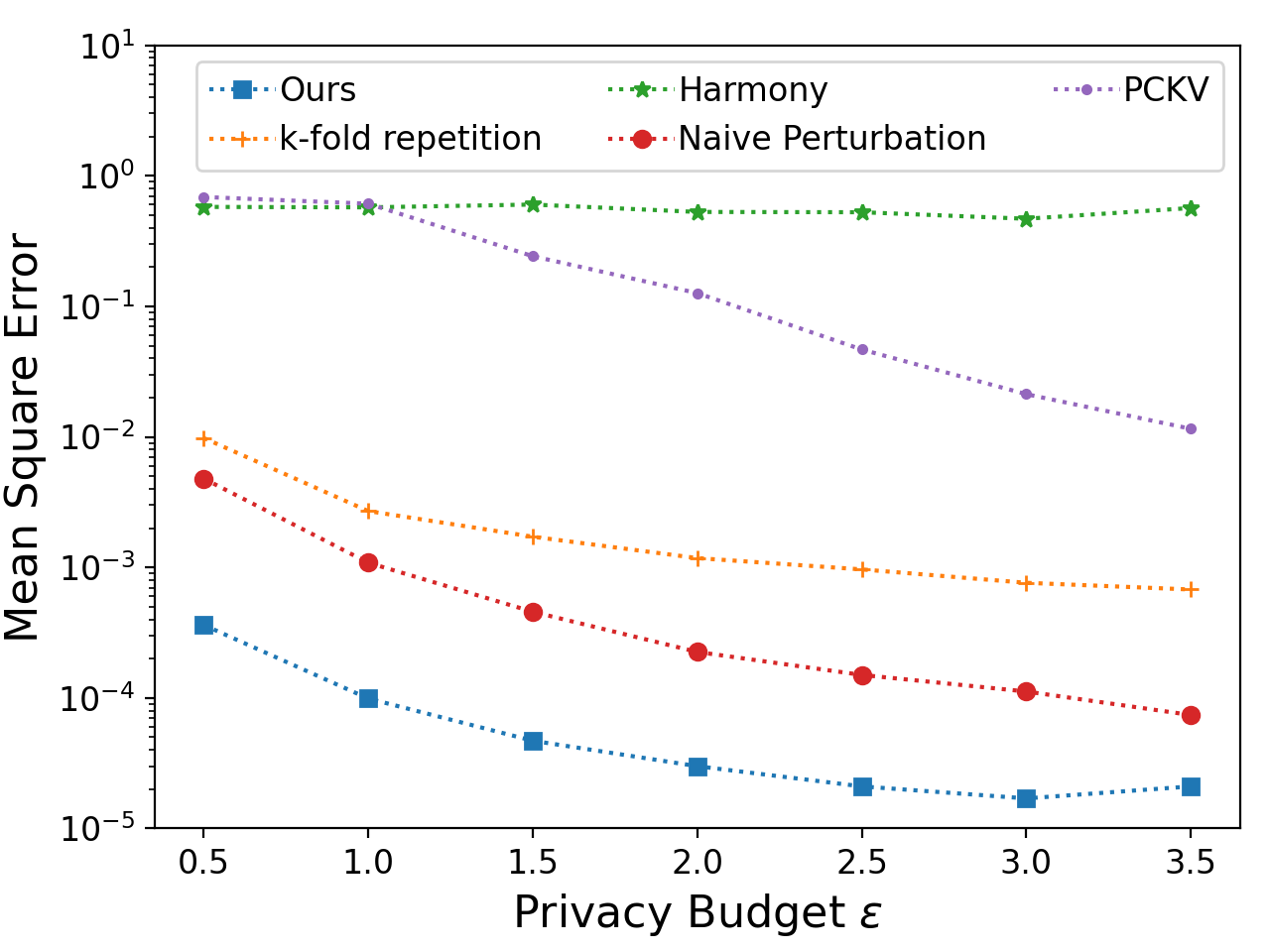}
    }
       \subfigure[MSE results for top 100 coordinates when fixing $n=10^5$, $k=64$, $\eps=1.0$ and varying $d$ from 64 to $10^5$.\label{fig:vary-d-mse-item}]{
    \includegraphics[width=0.31\textwidth]{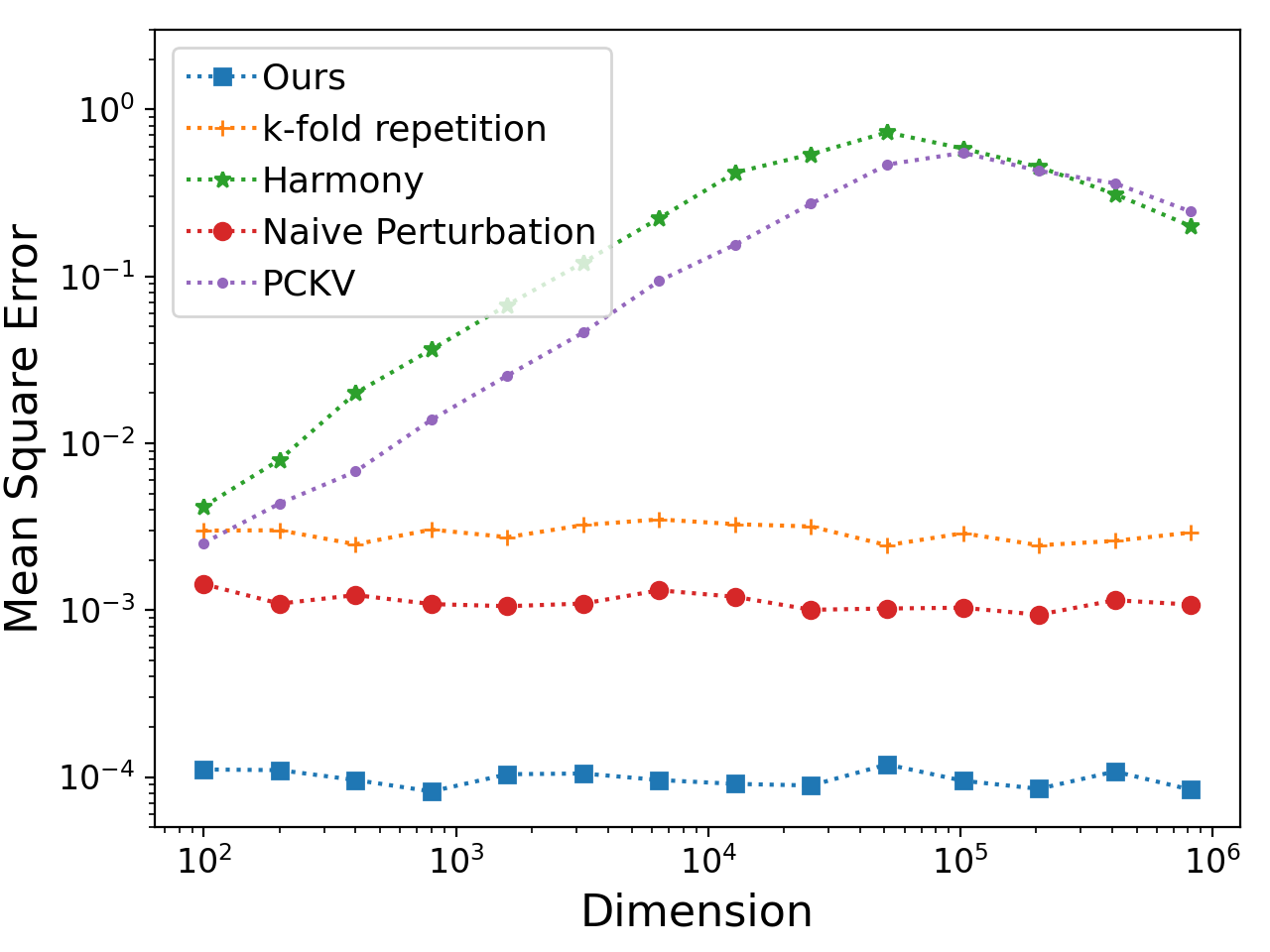}
    }
    \caption{Comparing the utilities of our method and existing approaches under event-level LDP. The plotting convention follows that of Figure~\ref{fig:main}: subfigures in the top and bottom row show the results of $L_\infty$ error and MSE, respectively. The subfigures in the left, middle, and right column varies sparsity $k$ (from 1 to 1024), $\epsilon$ (from 0.5 to 3.5), and dimension $d$ (from 64 to $10^5$), respectively (while fixing the other two parameters).}
  \label{fig:event-level}
\end{figure*}

The results are plotted in Figure~\ref{fig:event-level}. Theoretically (from Table~\ref{tab:comp1}), our method is better than other methods by at least a polynomial gap $\sqrt{k}$ in terms of the $L_\infty$ error. The following experiments verify the theoretical results.

\paragraph{Varying sparsity $k$.} The results are shown in Figure~\ref{fig:vary-k-mae-item} and Figure~\ref{fig:vary-k-mse-item}. The results matches our theoretical results that our methods are not scale with $k$ in event-level LDP. It also outperforms other methods for the whole range. 

\paragraph{Varying privacy budget $\eps$. }The results The results are shown in Figure~\ref{fig:vary-e-mae-item} and Figure~\ref{fig:vary-e-mse-item}. The higher privacy budget are beneficial to all methods' utility performance. Our algorithm still has the best estimation performance for the whole range. 

\paragraph{Varying dimension $d$. }The results are shown in Figure~\ref{fig:vary-d-mae-item} and Figure~\ref{fig:vary-d-mse-item}. Again, our algorithm has decoupled from the dimension $d$ and it has much smaller error estimation than other methods.

\subsection{Performance under $L$-Neighboring LDP}

\begin{figure}[]
    \centering
    \subfigure[\label{fig:vary-l-mae}]{
    \includegraphics[width=0.33\textwidth]{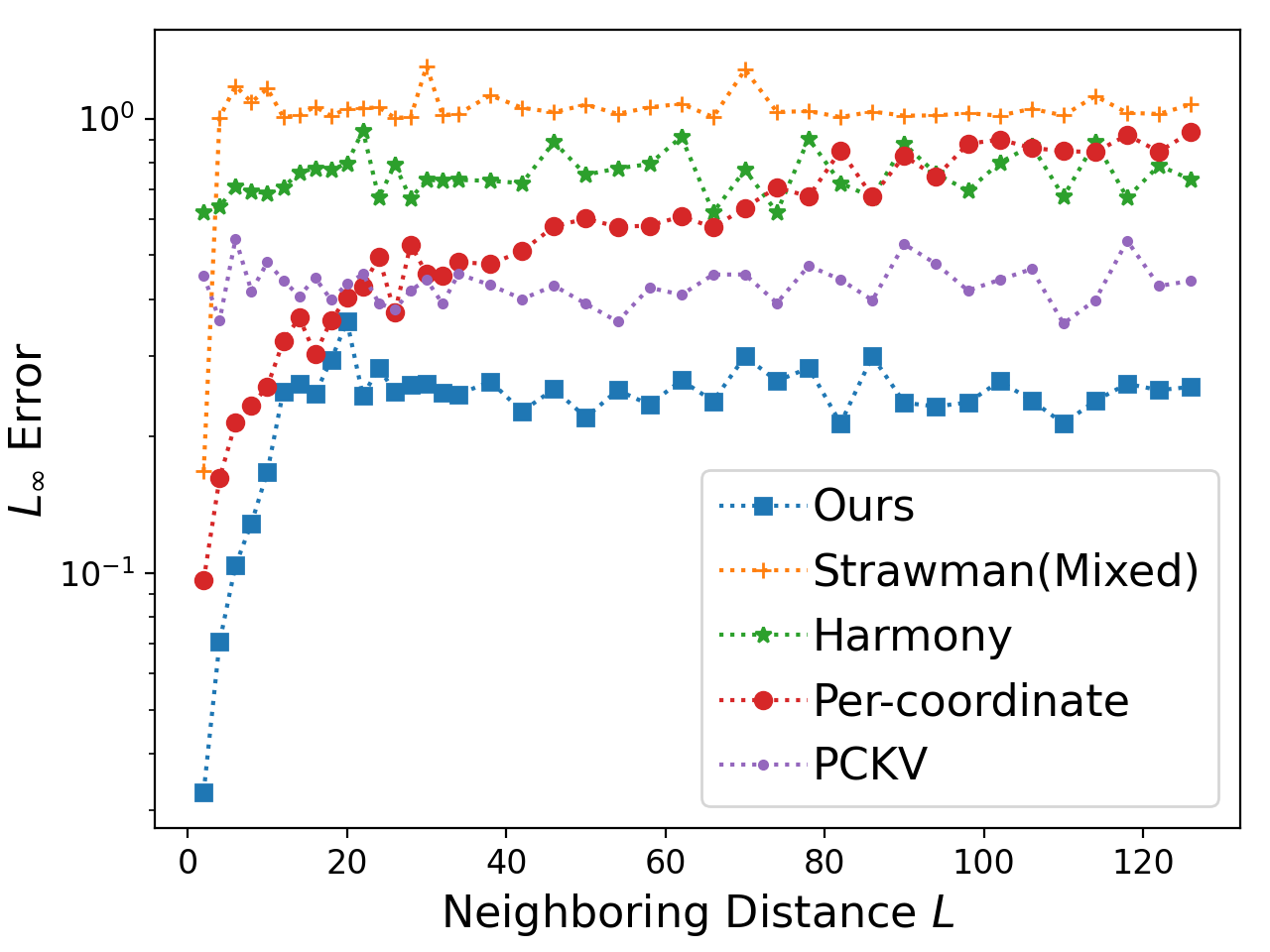}
    }
       \subfigure[\label{fig:vary-l-mse}]{
    \includegraphics[width=0.33\textwidth]{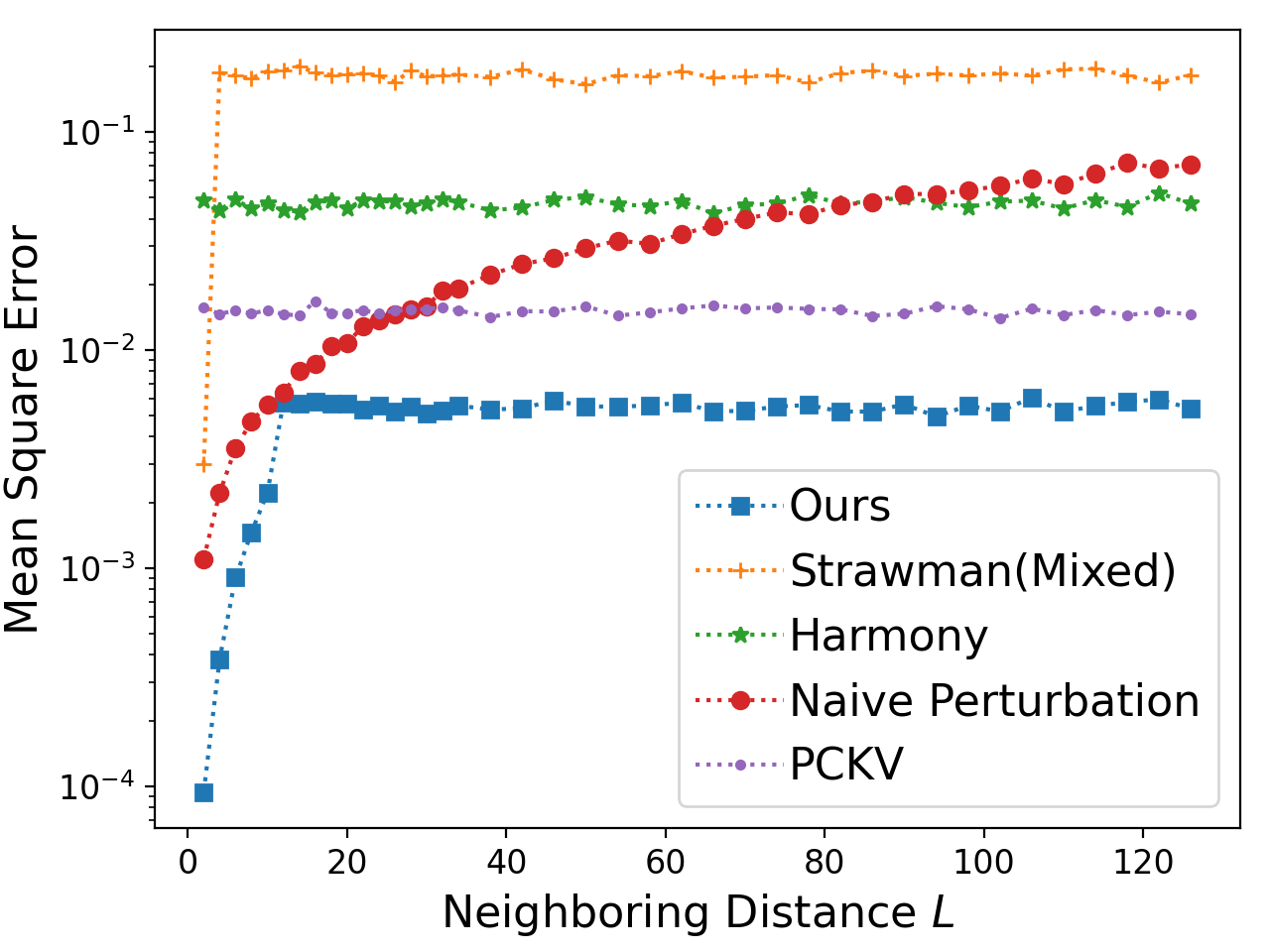}
    }
    \caption{Comparing the utilities of our method and existing approaches under $L$-Neighboring LDP. }
  \label{fig:L}
\end{figure}

The neighboring distance $L$ provides a better way to describe the middle ground between user-level LDP and event-level LDP. Our algorithm has theoretical $L_\infty$ error of $\min\{O(L), O((kL\log(1/\delta)^{\frac{1}{4}})),\\O(\sqrt{k\log n})\}\cdot O(\frac{1}{\eps}\sqrt{\frac{\log d}{n}})$. In the experiment, we fix the sparsity $k=64$ and vary the neighboring $L_1$ distance from $1$ to $128$. 
In the case when $L \ll k$, the parameter configuration with $O((kL\log(1/\delta))^{\frac{1}{4}})$ error growing factor should have asymptotically advantage over the configuration with $O(\sqrt{k\log n})$ growing factor.
However, in practice, we realize that the latter scheme(the algorithm with clipping) has a much smaller constant factor. 
Hence, in the case when $k$ is not large enough, we only see the optimized clipping scheme dominates the unclipped scheme.
The mixed strawman solutions, including $k$-fold repetition scheme and sampling scheme, can only adapt to either event-level LDP or user-level LDP. PCKV and Harmony cannot fully utilize the relaxed privacy as a way to improve the estimation error. The naive perturbation mechanism has worse scaling factor than our method, but in the turning point where $L=\sqrt{k \log n}$, it roughly matches the error of our method.

\subsection{Real-world Dataset Experiments}

\begin{table}[t]
\centering
\begin{adjustbox}{width=\linewidth}
\begin{tabular}{l|ccc|ccc}
\toprule
\multirow{2}{*}{\textbf{Name}} & \multicolumn{3}{c|}{\textbf{Event-level LDP}}                                                & \multicolumn{3}{c}{\textbf{User-level LDP}}                                                \\ 
                               & \multicolumn{1}{c|}{\textbf{Comm. Cost}} & \multicolumn{1}{c|}{\textbf{$L_\infty$ Err.} } & \textbf{MSE} & \multicolumn{1}{c|}{\textbf{Comm. Cost}} & \multicolumn{1}{c|}{\textbf{$L_\infty$ Err.}} & \textbf{MSE} \\ \midrule
$k$-fold repetition                       & \multicolumn{1}{c|}{28}                  & \multicolumn{1}{c|}{0.092}        & 0.00059      & -                   & -         & -       \\ 
\rowcolor{mygray}
Sampling                       & -            & -      & -  & \multicolumn{1}{c|}{8}                   & \multicolumn{1}{c|}{0.18}         & 0.0034       \\ 
Naive Perturbation                       & \multicolumn{1}{c|}{5512}                & \multicolumn{1}{c|}{0.17}         & 0.0023       & \multicolumn{1}{c|}{5512}                & \multicolumn{1}{c|}{0.41}         & 0.014        \\ 
\rowcolor{mygray}
Harmony                        & \multicolumn{1}{c|}{-}                   & \multicolumn{1}{c|}{-}            & -            & \multicolumn{1}{c|}{8}                   & \multicolumn{1}{c|}{1.0}          & 0.13         \\ 
PCKV                           & \multicolumn{1}{c|}{-}                   & \multicolumn{1}{c|}{-}            & -            & \multicolumn{1}{c|}{8}                   & \multicolumn{1}{c|}{0.70}         & 0.038        \\ 
\rowcolor{mygray}
Ours                           & \multicolumn{1}{c|}{9}                   & \multicolumn{1}{c|}{0.040}       & 0.00019      & \multicolumn{1}{c|}{9}                   & \multicolumn{1}{c|}{0.098}        & 0.00095      \\ \bottomrule
\end{tabular}
\end{adjustbox}
\hfill

\label{tab:cloth}
\caption{Experiment results for the Clothing Dataset.}
\end{table}

\begin{table}[t]
\centering
\begin{adjustbox}{width=\linewidth}
\begin{tabular}{l|ccc|ccc}
\toprule
\multirow{2}{*}{\textbf{Name}} & \multicolumn{3}{c|}{\textbf{Event-level LDP}}                                                & \multicolumn{3}{c}{\textbf{User-level LDP}}                                                \\ 
                               & \multicolumn{1}{c|}{\textbf{Comm. Cost}} & \multicolumn{1}{c|}{\textbf{$L_\infty$ Err.}} & \textbf{MSE} & \multicolumn{1}{c|}{\textbf{Comm. Cost}} & \multicolumn{1}{c|}{\textbf{$L_\infty$} Err.} & \textbf{MSE} \\ \midrule
$k$-fold repetition                       & \multicolumn{1}{c|}{48}                  & \multicolumn{1}{c|}{0.085}        & 0.00050      & -                  & -        & -    \\ 

\rowcolor{mygray} Sampling                       & -                & -    & -      & \multicolumn{1}{c|}{8}                   & \multicolumn{1}{c|}{0.29}         & 0.0052       \\ 
Naive Perturbation                       & \multicolumn{1}{c|}{23400}               & \multicolumn{1}{c|}{0.12}         & 0.0010       & \multicolumn{1}{c|}{23400}               & \multicolumn{1}{c|}{0.40}         & 0.012        \\ 
\rowcolor{mygray} Harmony                        & \multicolumn{1}{c|}{-}                   & \multicolumn{1}{c|}{-}            & -            & \multicolumn{1}{c|}{8}                   & \multicolumn{1}{c|}{1.0}          & 0.24         \\ 
PCKV                           & \multicolumn{1}{c|}{-}                   & \multicolumn{1}{c|}{-}            & -            & \multicolumn{1}{c|}{8}                   & \multicolumn{1}{c|}{1.0}          & 0.074        \\ 
\rowcolor{mygray} Ours                           & \multicolumn{1}{c|}{13}                  & \multicolumn{1}{c|}{0.033}        & 0.000083     & \multicolumn{1}{c|}{9}                   & \multicolumn{1}{c|}{0.11}         & 0.00091      \\ \bottomrule
\end{tabular}
\end{adjustbox}
\hfill

\label{tab:rent}
\caption{Experiment results for the Renting Dataset.}
\end{table}

\begin{table}[t]
\centering
\begin{adjustbox}{width=\linewidth}
\begin{tabular}{l|ccc|ccc}
\toprule
\multirow{2}{*}{\textbf{Name}} & \multicolumn{3}{c|}{\textbf{Event-level LDP}}                                                & \multicolumn{3}{c}{\textbf{User-level LDP}}                                                \\ 
                               & \multicolumn{1}{c|}{\textbf{Comm. Cost}} & \multicolumn{1}{c|}{\textbf{$L_\infty$ Err.}} & \textbf{MSE} & \multicolumn{1}{c|}{\textbf{Comm. Cost}} & \multicolumn{1}{c|}{\textbf{$L_\infty$ Err.}} & \textbf{MSE} \\ \midrule
$k$-fold repetition                       & 404                  & 0.25        & 0.0033      & -                  & -        & -    \\ 

\rowcolor{mygray} Sampling                       & -                & -    & -      & \multicolumn{1}{c|}{8}                   & \multicolumn{1}{c|}{1.0}         & 0.29       \\ 
Naive Perturbation                       & \multicolumn{1}{c|}{109052}               & 0.12         & 0.00082       & \multicolumn{1}{c|}{109052}               & \multicolumn{1}{c|}{1.0}         & 0.081        \\ 
\rowcolor{mygray} Harmony                        & \multicolumn{1}{c|}{-}                   & \multicolumn{1}{c|}{-}            & -            & \multicolumn{1}{c|}{8}                   & \multicolumn{1}{c|}{1.0}          & 0.49         \\ 
PCKV                           & \multicolumn{1}{c|}{-}                   & \multicolumn{1}{c|}{-}            & -            & \multicolumn{1}{c|}{8}                   & \multicolumn{1}{c|}{1.0}          & 0.24        \\ 
\rowcolor{mygray} Ours                           & 105                  & 0.034        & 0.000065     & \multicolumn{1}{c|}{9}                   & 0.25         & 0.0021      \\ \bottomrule
\end{tabular}
\end{adjustbox}
\hfill

\label{tab:movie}
\caption{Experiment results for the Movie Dataset.}
\end{table}

We compile the Clothing, Renting and Movie dataset to the sparse vector mean estimation problem. The description of the datasets can be found in Table~\ref{tab:rent}. Our method achieves best accuracy in both event-level LDP setting and user-level LDP setting by a magnitude of gap. Specifically, compared to our method, the strawman scheme has an extra $\sqrt{k}$ factor in the $L_\infty$ error, which is roughly 2.4, 3.3 and 8.0 in the three datasets correspondingly. The Harmony and PCKV schemes do not output very meaningful estimation in the experiments because their algorithms have error scaled with the dimension $d$.

%% file: tex/lower.tex
\section{Lower Bound} 
\label{sec:lower}

In a previous work~\cite{bassily2015local}, Bassily and Smith showed a lower bound of $\Omega\left(\frac{1}{\eps}\sqrt{\frac{\log d}{n}}\right)$ on the $L_{\infty}$ error under the 1-sparse case with the constraints of $(\eps, o(\frac{1}{n\log n}))$-LDP (Theorem \ref{theorem:one_item_lower_bound}). 
The 1-sparse case can be seen as a special case for the general $k$-sparse vector mean estimation under event-level LDP. 
Our algorithm for event-level LDP matches this lower bound, making the error bound tight in the event-level LDP case.

We observe that it is not hard to extend 
the framework and prove a lower bound of $\Omega\left(\frac{1}{\eps}\sqrt{\frac{k\log(d/k)}{n}}\right)$ on the $L_{\infty}$ error of $k$-sparse vector mean estimation under the user-level LDP. 
For completeness, we present the full proof below.

\paragraph{Notation.} 
In the lower bound proof,
each client~$i$ has a $k$-sparse input vector $v_i \in \mcal{S} := \{v \in \{0,1\}^d: \|v\|_1 = k \}$, where  
the special case $k=1$
is essentially the one-item frequency estimation problem~\cite{bassily2015local}.
Note that since this setting is a special case of real-valued mean vector estimation, the lower bound applies to mean vector estimation more generally.

Each client~$i$
applies an $(\epsilon, \delta)$-differentially private (where
any two inputs in $\mcal{S}$ are neighboring) algorithm $\mcal{Q}_i(\cdot)$
independently to produce  
$\bz_i = \mcal{Q}_i(v_i)$ in some report space~$\mcal{Z}$.  The server computes
$\widehat{v} := \mcal{A}(\bz_1, \ldots, \bz_n)$,
which estimates $\frac{1}{n} \sum_{i} v_i$.
Then the following lower bound holds.

\begin{theorem}[Lower bound on error, $k$-sparse mean vector estimation]
\label{th:main_lb}
Let $0 < \epsilon = O(1)$ and $0 < \delta = o(\frac{\epsilon}{n \log n})$.
Suppose for each client~$i$,
the (randomized) algorithm~$\mcal{Q}_i: \mcal{S} \rightarrow \mcal{Z}$ is $(\epsilon, \delta)$-differentially
private, where any two inputs in $\mcal{S}$ are considered as neighboring.
Moreover, $\mcal{A}: \mcal{Z}^n \rightarrow \mcal{S}$ is a (potentially randomized) aggregator function.

Then, there exists some distribution $\mcal{P}$ on $\mcal{S}$ (depending on $\mcal{Q}_i$'s and $\mcal{A}$) such that
if every client~$i$ independently generates a report $\bz_i = \mcal{Q}_i(v_i)$, 
where $v_i$ is sampled from $\mcal{P}$ independently, the expected error of estimating $\overline{v} := \frac{1}{n}\sum_i v_i$  has the following lower bound:

$\E[ \| \mcal{A}(\bz_1, \ldots, \bz_n) - \overline{v} \|_\infty] \geq
\min \left\{\Omega\left(\frac{1}{\eps}\sqrt{\frac{\log |\mcal{S}|}{n}}\right), 1\right\}$,

where $\log |\mcal{S}| = \log {d \choose k} \geq k\log(d/k)$.
\end{theorem}

Plugging in the definition of $\mcal{S}$, the following corollary gives our main lower bound
for user-level-LDP.
\begin{corollary}
Observing that if $L = 2k \leq \sqrt{d}$, then
any two inputs in $\mcal{S}$ has $L_{\infty}$ distance
at most $L$.
Hence, in this case,
$\E[ \| \mcal{A}(\bz_1, \ldots, \bz_n) - \overline{v} \|_\infty] \geq
\min \{\Omega(\frac{1}{\eps}\sqrt{\frac{L \log d}{n}}), 1\}$
\end{corollary}

\paragraph{Proof roadmap.}
Just like Bassily and Smith,  
our goal is to find a ``hard'' joint distribution $\cP$ on the clients' inputs $\v=(v_1,\dots,v_n)$, such that the expected $L_\infty$ estimation error is large for any $(\eps,\delta)$-user-level LDP algorithm, Here, the expectation is taken over the randomness coming from the input sampling and the algorithm. 
We construct the distribution $\cP$ as following. 
First, a vector $V\in \{0,1\}^d$ is sampled uniformly at random from a candidate set $\mcal{S}$ that includes all binary $k$-sparse vectors in $\{0,1\}^d$. Next, each client's input $v_i$ is sampled i.i.d from a distribution $\cP_{V}^{(\eta)}$ (using the same~$V$ for all users) as follows: 
\begin{eqnarray}
	v_i =
	\begin{cases}
		V & w.p. ~~ \eta \\
		U & w.p. ~~ 1-\eta 
	\end{cases}
	\label{eq:degrading}
\end{eqnarray}
where $U$ is drawn uniformly from $\cS$.
The distribution $\cP_{V}^{(\eta)}$ is an instance of an \textit{$\eta$-degrading channel}\cite{bassily2015local}.
To prove that any $(\eps,\delta)$-LDP algorithm has large error with respect to $\cP_{v}^{(\eta)}$ for at least one $v \in \mcal{S}$, we view the problem as an encoding-decoding process, then bound the error using Fano's inequality.
Each client $i$ generates a report $\bz_i=\Q_i(v_i)$.
The joint reports $\bz := (\bz_1, \bz_2, \ldots, \bz_n)$ are viewed as a noisy encoding of $V$. In an attempt to recover~$V$,
the server aggregator function~$\mcal{A}$ is applied to produce the mean estimation $\mcal{A}(\bz)$. 
Then, to decode the original $V$, the server removes the bias introduced by the degrading channel then rounds the estimation $\mcal{A}(\bz)$ to the nearest binary vector $\widehat{V}$. 
A decoding error occurs if $\widehat{V} \neq V$.

The lower bound proof relies on two bounds on the
probability of decoding error.  On one hand, the differential privacy
of each $\mcal{Q}_i$ implies that the mutual information $I(V;\bz)$ is small, which means that
the decoding error probability is large by Fano's inequality.
On the other hand, a small $L_\infty$-error estimation of
the mean vector implies that the original $V$ can be recovered
from $\mcal{A}(\bz)$ with high probability.
These effects limit the decoding error probability and
give us a lower bound on the mean vector estimation error.
That is, for small enough $\delta = o(\frac{\epsilon}{n \log n})$,
we find a distribution $\cP$ over candidate set $\mcal{S}$ that implies 
a lower bound of $\Omega(\frac{1}{\epsilon}\sqrt{\frac{\log |\mcal{S}|}{n}})$
on the $L_\infty$-error of mean vector estimation.
By considering $k$-sparse vectors in $\{0,1\}^d$,
(for which $|\mcal{S}| = {d \choose k}$), we obtain the lower bound of  $\Omega\left(\frac{1}{\epsilon}\sqrt{\frac{k\log (d/k)}{n}}\right)$.

In the interest of space, we defer the detailed
lower bound proof to Appendix~\ref{sec:appendix-lb}.

%% file: tex/prelim-more.tex
\section{Additional Preliminaries}

\begin{theorem}[Sequential Composition Theorem]
\label{thm:seq_comp}
	Assume the distribution of $X$ and $X'$ are $(\eps_1,\delta_1)$-close. If for any of $x\in Domain(X)$, the posterior distribution of random varaible $Y$ conditioned on $X=x$ and random variable $Y'$ conditioned on $X'=x$ are $(\eps_2,\delta_2)$-close, then the distribution of $(X,Y)$ and $(X',Y')$ are $(\eps_1+\eps_2,\delta_1+\delta_2)$-close. 
\end{theorem}


\begin{theorem}[Post Processing Theorem]
	Assume the distribution of $X$ and $X'$ are $(\eps,\delta)$-close. Then for any (randomized) function $f$, the distributions of $f(X)$ and $f(X')$ are $(\eps, \delta)$-close. 
\end{theorem}



%% file: tex/appendix.tex

\section{Additional Details of our Upper Bound Construction}
\subsection{Discretization of Real Values for Communication}
\label{sec:communication}

The clients need to send the reports tuple $(h_i, s_i, \tB_{i,1},\dots,tB_{i,b})$ to the server. 
For the $h_i$ and $s_i$, the client can send the random seed for the PRF to the server and the communication cost is $O(\xi)$. 
Here, $\xi$ is the security parameter and with only ${\sf negl}(\xi)$ probability, the randomness will be broken.
The bucket values are unbounded real values.
We actually know the ``raw bucket value'' $B_{i,1},\dots,B_{i,b}$ are trivially bounded by $[-k,k]$. 
The unbounded part comes from the Laplacian noise and it has good concentration property.
We can clip the value again with range $[-U,U]$, where $U=k+(\frac{\lambda}{\eps})\log \frac{10nb}{\beta}$.
We know if a random variable $r\sim \Lap(\frac{\lambda}{\eps})$, then $\Pr[|r| \ge t ] \le e^{-t\eps / \lambda}$. 
Using union bound, we know that the magnitudes of all $nb$ Laplacian random variables are smaller than $(\frac{\lambda}{\eps})\log \frac{10nb}{\beta}$ with prob. $1-\beta/10$. Then, we know for all clients, the report values are bounded by $[-U,U]$ with prob. $1-\beta/10$. 
Then, we discretize the value using the unbiased discretizer

\begin{align*}
    \label{dsc}
		{\tt DSC}(x)=
		\begin{cases}
			\lfloor x \rfloor, \text{  w.p. } 1-(x-\lfloor x \rfloor) \\
			\lfloor x \rfloor + 1, \text{  w.p. } x-\lfloor x \rfloor \\
		\end{cases}
\end{align*}

\noindent We need to prove the error introduced by discretization is small. 
We denote the discretized version of $\tB^c_{i}$ as $\textbf{B}^c_{i}$.
Trivially, for $c\in[n]$, $|\tilde{B}_{h^c(x)}s^c(x)- \textbf{B}^c_{i}s^c(x)| \le 1$, and also $\E[\tilde{B}^c_{h^c(x)}]=\E[\textbf{B}^c_{i}]$. 
Using Hoeffding's inequality, we can prove that $|\frac{1}{n}\sum_{c\in[n]}\tilde{B}^c_{h^c(x)}s^c(x)-\frac{1}{n}\sum_{c\in[n]} \textbf{B}^c_{i}s^c(x)| \le O\left(\sqrt{\frac{\log (d/\beta)}{n}}\right)$ with prob. at least $1-\beta/10d$. 
Combining the above argument, we conclude that error introduced by the communication process is asymptotically equal or less than the error introduced by other process.
	
	

\subsection{Additional Details for the Privacy Proof}
\label{sec:privacy-proof}

\begin{claim}[Restatement of Claim\ref{claim:lap-noise}]
Given any two neighboring vectors $v,v'$. 
Taking the randomness of $h$ and $s$,
if $\Pr[\sum_{j\in[b]}|B_j - B'_j| \ge \lambda]\le \delta$,
then the two invocation of the local randomizers' output
distributions are $(\eps,\delta)$-close.
\end{claim}

\begin{proof}

    \ignore{
    Fix any two $k$-sparse vector $v,v'\in[-1,1]^d$ such that $\Vert v-v'\Vert_1 \le L$. We firstly consider a simple case where $L \le 2\eta b$, so that the noise parameter is $\lambda=L/\eps$. Notice that the distribution of the hash function $h$ and the sign function $s$ are not dependent on the input vector $\v$. After the binning process and before the clipping process, the two vectors' information is ``mapped'' into the buckets. We say the values are $B_1,\dots,B_b$ for vector $v$ and $B'_1,\dots, B'_b$ for vector $v'$. We now want to prove that, for any $h,s$, $\sum_{i\in[b]}|B_i-B'_i|\le L$. 
    \begin{align*}
          & \sum_{i\in[b]}|B_i-B'_i| \\ 
        = & \sum_{i\in[b]}\left|\sum_{j\in[d], h(j)=i} s(j)(v_i-v'_i)\right| \\ 
        \le & \sum_{i\in[b]}\sum_{j\in[d], h(j)=i} |v_i-v'_i| \\
        = & L
    \end{align*}
    }
    
    The clipping process will only make the difference smaller.
    Conditioned on the case that $\sum_{j \in [b]} |B_{j} - B'_j|\le \lambda$. Then, we can bound the ratio between the probability density function of the r.v. $B_j+r_j$ and $B'_j + r_j$ : for any $x_1, \dots, x_b \in \mathbb{R}$,
    \begin{align*}
        & \frac{p_v(\land_{j \in b}B_j+r_j=x|h,s)}{p_{v'}(\land_{j \in b}B'_j+r_j=x|h,s)}  
        =  \frac{\prod_{j\in[b]}\exp(\frac{\eps}{\lambda}\|x-B_j|)}{\prod_{j\in[b]}\exp(\frac{\eps}{\lambda}|x-B'_{j}|)} \\
        = & \exp\left(\frac{\eps}{\lambda} \sum_{j \in [b]} (|x-B_j|-|x-B'_j|)\right) \\
        \le & \exp\left(\frac{\eps}{\lambda} \sum_{j \in [b]} |B_j-B'_j|\right) 
        \le  \exp(\eps).
    \end{align*}
    
    Therefore, the distribution of $\tilde{B}_i$ and $\tilde{B}'_i$ is $(\eps,0)$ close. 
    By post-processing theorem, the joint distribution of the final report $\tB_1,\dots,\tB_b$ and $\tB'_1,\dots,\tB'_b$ is still $(\eps,0)$-close.
    Considering the failure probability $\delta$ such that
    some randomly sampled $h$ and $s$ cause
    $\sum_{j\in[b]}|B_j-B'_j|> \delta$, 
    the distributions for the whole outputs $(h, s, \tB_1,\dots,\tB_b)$ and $(h, s, \tB'_1,\dots,\tB'_b)$ are $(\eps,\delta)$-close.

    
    
\end{proof}

\subsection{Full Proof of the Utility Theorem}
\label{sec:full-proof}

Below we give the full proof of the utility statement in Theorem~\ref{thm:main-utility}.

\begin{proof}

    
	Our proof bound the absolute error incurred step by step. 
	Fix an index $x\in[d]$.
    The server computes the estimation
    $\hat{v}_x$ as $\hat{v}_x = \frac{1}{n}\sum_{i\in[n]} s_i(x)\tB_{i, h_i(x)}$. 
    We bound the error by the three steps: 1) binning; 2) clipping; 3) adding Laplacian noise.
	
	\begin{align*}
		& |\bar{v}_x-\hat{v}_x| \\
        = & \left|\frac{1}{n}\sum_{i\in[n]} v_{i,x} - \frac{1}{n}\sum_{i\in[n]} s_i(x)\tB_{i, h_i(x)}\right| \\ 
		\le & \left|\frac{1}{n}\sum_{i\in[n]} v_{i,x} - \frac{1}{n}\sum_{i\in[n]} B_{i, h_i(x)}s_i(x)\right| \\
		& +\left|\frac{1}{n}\sum_{i\in[n]} B_{i, h_i(x)}s_i(x)-\frac{1}{n}\sum_{i\in[n]} \bar{B}_{i, h_i(x)}s_i(x)\right|  \\
        & +\left|\frac{1}{n}\sum_{i\in[n]} \bar{B}_{i, h_i(x)}s_i(x)-\frac{1}{n}\sum_{i\in[n]}\tilde{B}_{i, h_i(x)}s_i(x)\right|  \\
	\end{align*}
	
	We now look at the binning error term 
	$ |\frac{1}{n}\sum_{i\in[n]} v_{i,x} - \frac{1}{n}\sum_{i\in[n]} B_{i,h_i(x)}s_i(x)| $. 
	We have $B_{i, h_i(x)}=\sum_{l\in[d]}\I[h_i(l)=h_i(x)]v_{i,l}s_i(l)$. 
	Define random variables
	$Y_{i,l}=\I[h_i(l)=h_i(x)]v_{i,l}s_i(l)s_i(x)$ for $i\in[n],l\in[d]$. 
	Then the error term can be rewritten as
	$|\frac{1}{n}\sum_{i\in[n]}\sum_{l\in[d],l\ne x}Y_{i,l}|$. 
	Since $s_i(\cdot)$ is a random $\pm 1$ function, $s_i(l)s_i(x)$ can be seen as an independent uniform $\pm 1$ random variable.
	Also, the hash function $h_i$ is a random oracle, so for $l\ne x$,
	$\Pr[\I[h_i(l)=h_i(x)]=1]=\frac{1}{b}$. 
	So we know $\E[Y_{i,l}]=0,\E[(Y_{i,l})^2]=\frac{1}{b}v_{i,l}^2 \le \frac{1}{b}$. 
	Thus, we can use the Bernstein's Inequality:
	\begin{align*}
	    & \Pr\left[\left|\frac{1}{n}\sum_{i\in[n]}\sum_{l\in[d],l\ne x}Y_{i,l}\right| \ge t\right] \\
	    \le & 2\exp\left(-\frac{(nt)^2}{\sum_{i\in[n]}\sum_{l\in[d],l\ne x}\E[(Y_{i,l})^2]+\frac{1}{3}nt}\right) \\
	    \le & 2\exp\left(-\frac{(nt)^2}{\frac{nk}{b}+\frac{1}{3}nt}\right) 
	    =  2\exp\left(-\frac{nt^2}{\frac{k}{b}+\frac{1}{3}t}\right).
	\end{align*}
	
	Using the fact that $k/b\ge \log(5d/\beta)/n$ and setting $t=\Theta(\sqrt{\frac{k}{b}}\sqrt{\frac{\log(d/\beta)}{n}})$ with a proper constant, we can prove that $\Pr[|\frac{1}{n}\sum_{i\in[n]}\sum_{l\in[d],l\ne x}Y_{i,l}| \ge t] \le \beta/10d$. That means, $ |\frac{1}{n}\sum_{i\in[n]} v_{i,x} - \frac{1}{n}\sum_{i\in[n]} B_{i, h_i(x)}s_i(x)| =O(\sqrt{\frac{k}{b}}\sqrt{\frac{\log(d/\beta)}{n}})$ with probability $1-\beta/10d$.
	
	Next, we look at the clipping error term $|\frac{1}{n}\sum_{c\in[n]} B_{i, h_i(x)}s_i(x)-\frac{1}{n}\sum_{i\in[n]} \bar{B}_{i, h_i(x)}s_i(x)|$. We actually try to prove the clipping range is large enough, so that, with probability $1-\beta/2$, for all $i\in[n],j\in[b]$, $B_{i,j}=\bar{B}_{i,j}$, i.e., $|B_{i,j}|\le \eta$. Then, the error term becomes zero naturally. Fix any $i,j$. $B_{i,j}=\sum_{l\in[d]} \I[h_i(l)=j]s_i(l)v_{i,l}$. We define a random variable $X_l=\I[h_i(l)=j]s_l(l)v_{i,l}$. Then we know $B_{i,j}=\sum_{l\in[d]} X_l$.  Its moment generating function is 
    \begin{align*}
       & \E[\exp(sB_{i,j})] =  \prod_{l\in[d]} \exp(sX_l) \\
        = & \prod_{l\in[d]} \left(1-\frac{1}{b}+\frac{1}{2b}(e^{-sv_{i,l}} + e^{sv_{i,l}})\right) \\
        \le & \prod_{l\in[d]} \frac{1}{2}(e^{-sv_{i,l}} + e^{sv_{i,l}})  
        \le  \prod_{l\in[d]} \exp\left(s^2v^2_{i,l}/2\right) \\
        = & \exp\left(\frac{s^2}{2}\sum_{l\in[d]} v^2_{i,l}\right)  
        \le  \exp\left(\frac{ks^2}{2}\right). 
    \end{align*}

   Hence, $B_{i,j}$ is a sub-Gaussian random variable with variance $k$.
    When $\eta\ge\sqrt{2k\log(4nb/\beta)}$,  $\Pr[|B_{i,j}|\ge \eta]\le 2\exp(-\frac{\eta^2}{2k})=2\exp(-\log(4nb/\beta))=\beta/2nb$. Using union bound, we prove that, with prob. $1-\beta/2$, for all $i\in[n],j\in[b]$, $|B_{i,j}|\le \eta$, i.e., $\bar{B}_{i,j}=B_{i,j}$.
    
	Now, we look at the error term, $|\frac{1}{n}\sum_{i\in[n]} \bar{B}_{i,h_i(x)}s_i(x)-\frac{1}{n}\sum_{i\in[n]}\tilde{B}_{i, h_i(x)}s_i(x)|$.
	We know that for all $j\in[b],i\in[n]$, $\tilde{B}_{i,j}=\bar{B}_{i,j}+r_{i,j}$, where $r_{i,j}\sim {\sf Lap}(\frac{\lambda}{\eps})$.
	Also, using the fact that Laplacian distribution is symmetrical over positive value and negative value and $s_i(x)$ is a uniform random $\pm 1$ variable, the error term can be rewritten as $|\frac{1}{n}\sum_{i\in[n]} r_{i,h_i(x)}|$. We know that the Laplacian noise $\Lap(\frac{\lambda}{\eps})$'s distribution is a sub-exponential distribution $\subE(\frac{\lambda}{\eps})$. Using the concentration bound for sub-exponential random variable, we can prove that with prob. $1-\beta/10d$, $|\frac{1}{n}\sum_{i\in[n]} \bar{B}_{i,h_i(x)}s_i(x)-\frac{1}{n}\sum_{i\in[n]}\tilde{B}_{i,h_i(x)}s_i(x)|\le O\left(\frac{\lambda}{\eps}
   \sqrt{\frac{\log(d/\beta)}{n}}\right)$.
   
   Taking the union bound over all $x\in[d]$, we prove that with probability at least $1-5\beta/4$,  $\max_{x \in [d]}|\frac{1}{n}\sum_{i\in[n]} v_{i,x} - \frac{1}{n}\sum_{i\in[n]} s_i(x)\tB_{i,h_i(x)}| \le O\left(\left(\sqrt{\frac{k}{b} }+\frac{\lambda}{\eps}\right)\sqrt{\frac{\log(d/\beta)}{n}}\right)$.

\end{proof}

%% file: tex/lb-detail.tex
\section{Detailed Lower Bound Proof}
\label{sec:appendix-lb}

In this section, we give the detailed proof of Theorem~\ref{th:main_lb}.

\paragraph{Estimation of distribution mean.}
Because the empirical average $\overline{v}$
is concentrated around the distribution mean $\E[\mcal{P}]$ (Lemma~\ref{lem:empdistr}),
it suffices to consider the expected error of estimating $\E[\mcal{P}]$
by the following quantity:

\begin{align}
    \cE(\cA; \cP) \triangleq \E[\linf{\cA(\bz_1,\bz_2,\dots,\bz_n)-\E[\cP]}], 
\end{align}
where the randomness comes from sampling $v_i$ from $\cP$, the randomized algorithms $\Q_i$'s from all clients and the estimator $\cA$.
The following result (which is also used in~\cite{bassily2015local}) implies that it suffices
to prove the same asymptotic lower bound for $\cE(\cA; \cP)$
to achieve Theorem~\ref{th:main_lb}.

\begin{lemma}[Empirical Average vs Distribution Mean]
\label{lem:empdistr}
Let $\overline{V}$ be the empirical average
of $n$ i.i.d. samples from $\cP$.
Then, $\E[\linf{\bar{V}-\E[\cP]}] \leq O\left(\sqrt{\frac{\log d}{n}}\right)
= o\left(\frac{1}{\eps}\sqrt{\frac{k\log(d/k)}{n}} \right)$.
\end{lemma}

\ignore{
\begin{proof}
Notice that each coordinate of the input vector is bounded by $[-1,1]$, so we can simply apply Hoeffding's Inequality and union bound. We get for every $t>0$, $\Pr[\linf{\bar{V}-\E[\cP]} \ge t] \le 2d\exp(\frac{-nt^2}{2})$. We set an appropriate $t=2\sqrt{\frac{\log (2d)}{n}}$. Finally, 
    \begin{align*}
        \E[\linf{\bar{V}-\E[\cP]}] & =  \int_{t=0}^{\infty} \Pr[\linf{\bar{V}-\E[\cP]} \ge t] dt\\
        & \le \sum_{i=1}^{\infty} t\Pr[\linf{\bar{V}-\E[\cP]} \ge (i-1)t] \\
        & = t (1 + \sum_{i=1}^{\infty} \Pr[\linf{\bar{V}-\E[\cP]} \ge it]) \\
        & = t (1 + \sum_{i=1}^{\infty} (2d)^{-2i}) \\
        & = O(t) \\
        & = O\left(\sqrt{\frac{\log (d)}{n}}\right) \\
    \end{align*}
\end{proof}
}

\paragraph{Analyzing expected error via an encoding-decoding process.}
Given randomized algorithms $\mcal{Q}_i$'s and aggregator function~$\mcal{A}$,
the lower bound framework in~\cite{bassily2015local}
consider the following encoding-decoding process.
Denote $\eta := \min \{\Theta(\frac{1}{\eps}\sqrt{\frac{\log |\mcal{S}|}{n}}), 1\}$.

\begin{enumerate}

\item Sample $V$ uniformly at random from $\mcal{S}$.

\item Each client~$i$ receives the same~$V$ from the previous step,
and performs the following actions independently.

\begin{itemize}

\item Sample $v_i$ from $\mcal{P}^{(\eta)}_V$, where for $v \in \mcal{S}$,
the distribution $\mcal{P}^{(\eta)}_v$ is defined as in \eqref{eq:degrading}.


\item Apply local LDP mechanism $\mcal{Q}_i(\cdot)$ to obtain $\bz_i := \mcal{Q}_i(v_i)$.
\end{itemize}

\item Using the aggregator function $\mcal{A}(\cdot)$,
compute $Y = \frac{1}{\eta} ( \mcal{A}(\bz_1, \ldots, \bz_n) - (1- \eta)
\cdot \frac{1}{n} \sum_{v \in \mcal{S}} v)$.

Round $Y$ to $\widehat{V} \in \{0,1\}^d$,
i.e., for each $j \in [d]$,
$\widehat{V}_j := 1$ if $Y_j \geq \frac{1}{2}$, and 0 otherwise.

\item Define the event $\mathsf{error}$ as $V \neq \widehat{V}$.

\end{enumerate}

The crux of the proof depends on the following bounds on $\Pr[\mathsf{error}]$:
\begin{itemize}
\item A lower bound by
Fano's Inequality: 

\begin{equation} \label{eq:Fano}
\Pr[\mathsf{error}] \geq 1 - \frac{I(V; \bz_1, \ldots, \bz_n) + 1}{\log |\mcal{S}|}.
\end{equation}

Since conditioning on $V$, the $\bz_i$'s are independent,
we have:
$I(V; \bz_1, \ldots, \bz_n) = \sum_{i \in [n]} I(V; \bz_i)$. 


An upper bound on
$I(V; \bz_i) = I(V; \mcal{Q}_i(v_i))$
using the differential privacy of $\mcal{Q}_i$ will be given in Lemmas~\ref{lem:mutual} and~\ref{lemma:mutual2}.

\item An upper bound on $\Pr[\mathsf{error}]$ is given in
Lemma~\ref{lemma:decoding_error}.
\end{itemize}

\begin{lemma}[Low Decoding Error]
\label{lemma:decoding_error}
Suppose for all $v \in \mcal{S}$,
$\cE(\cA; \cP^{(\eta)}_v) \leq \frac{\eta}{10}$.
Then, $\Pr[\mathsf{error}] \leq \frac{1}{5}$.
\end{lemma}

\begin{proof}
Observe that the event $\mathsf{error}$
implies that for at least one $j \in [d]$,
the difference between the $j$-th coordinates of $Y$ and $V$ is at least $\frac{1}{2}$,
i.e., $\|Y - V \|_\infty \geq \frac{1}{2}$.

Hence, by Markov's inequality,
the probability of this event is at most
$2 \cdot \E[\|Y - V \|_\infty ]$.
Finally, as shown in~\cite{bassily2015local}, observe that:
$\E[\|Y - V \|_\infty ] = \frac{1}{\eta} \cdot \E_V[\cE(\cA; \cP^{(\eta)}_V)]
\leq \frac{1}{10}$, which gives the result.
\end{proof}

\paragraph{Bounding mutual information via differential privacy.}
The following lemmas from~\cite{bassily2015local}
give upper bounds on the mutual information between
the input and the output of differentially private algorithms.

\begin{lemma}
    \label{lem:mutual}
		Suppose $0 < \epsilon = O(1)$ and $0 < \delta < \epsilon$.
    Let $V$ be a random variable that is uniformly distributed on a discrete set $\cS$. Suppose the output of the (randomized) algorithm $\Q: \cS \to \cZ$ is $(\eps,\delta)$-differentially private where
		any two inputs in $\cS$ are considered as neighboring. 
		Then, the mutual information between the input variable and the output report is bounded:  
    \begin{align*}
        I(V; \mcal{Q}(V))=O(\eps^2 + \frac{\delta}{\eps}\log|S| + \frac{\delta}{\eps}\log(\eps/\delta))
    \end{align*}
\end{lemma}

\begin{lemma}
\label{lemma:mutual2}
    Suppose $0 <\eps=O(1)$ and $0 < \delta < 1$
		and $\mcal{Q}: \mcal{S} \rightarrow \cZ$
		is $(\eps,\delta)$-differentially private.
			Define the algorithm $\mcal{Q}^{(\eta)}:
		\mcal{S} \rightarrow \cZ$ as follows:
		on input $v \in \mcal{S}$, sample $V$ from
		$\mcal{P}^{(\eta)}_v$ (which is
		defined in the encoding-decoding procedure) and return $\mcal{Q}(V)$.
		Then, the output of $\mcal{Q}^{(\eta)}$ is 
		$(O(\eta \eps), O(\eta \delta))$-differentially private.
\end{lemma}

\paragraph{Finalizing the proof of Theorem~\ref{th:main_lb}.}
For the sake of contradiction,
we assume that for any distribution
$\mcal{P}$ on $\mcal{S}$,
$\cE(\cA; \cP) \leq \frac{\eta}{10}$.
Then, Lemma~\ref{lemma:decoding_error}
implies that decoding error happens
with $\Pr[\mathsf{error}] \leq \frac{1}{5}$.

In view of Fano's Inequality~(\ref{eq:Fano}),
a contradiction can be achieved if
$\frac{I(V; \bz_1, \ldots, \bz_n) + 1}{\log |\mcal{S}| } \leq \frac{1}{2}$.

By Lemmas~\ref{lem:mutual} and~\ref{lemma:mutual2},
for each~$i$, $I(V; \bz_i) \leq O(\eta^2\eps^2 + \frac{\delta}{\eps}\log|S| + \frac{\delta}{\eps}\log(\eps/\delta))$.

By choosing sufficiently small $\eta = \min \{ \Theta(\frac{1}{\epsilon} \sqrt{\frac{\log |S|}{n}}, 1 \}$ and $\delta = o(\frac{\epsilon}{n \log n})$,
it follows that:

$\frac{I(V; \bz_1, \ldots, \bz_n) + 1}{\log |\mcal{S}| }
= \frac{\sum_{i \in [n]} I(V;\bz_i) + 1}{\log |\mcal{S}|}
< \frac{1}{2}$,
where the first equality holds
because conditioning on $V$, the $\bz_i$'s are independent.

Hence, we have obtained the desired contradiction
that completes the proof of Theorem~\ref{th:main_lb}.

%% file: single-col.bbl
\begin{thebibliography}{10}
\providecommand{\url}[1]{#1}
\csname url@samestyle\endcsname
\providecommand{\newblock}{\relax}
\providecommand{\bibinfo}[2]{#2}
\providecommand{\BIBentrySTDinterwordspacing}{\spaceskip=0pt\relax}
\providecommand{\BIBentryALTinterwordstretchfactor}{4}
\providecommand{\BIBentryALTinterwordspacing}{\spaceskip=\fontdimen2\font plus
\BIBentryALTinterwordstretchfactor\fontdimen3\font minus
  \fontdimen4\font\relax}
\providecommand{\BIBforeignlanguage}[2]{{%
\expandafter\ifx\csname l@#1\endcsname\relax
\typeout{** WARNING: IEEEtranS.bst: No hyphenation pattern has been}%
\typeout{** loaded for the language `#1'. Using the pattern for}%
\typeout{** the default language instead.}%
\else
\language=\csname l@#1\endcsname
\fi
#2}}
\providecommand{\BIBdecl}{\relax}
\BIBdecl

\bibitem{flworkshop}
\BIBentryALTinterwordspacing
``2021 workshop on federated learning and analytics.'' [Online]. Available:
  \url{https://events.withgoogle.com/2021-workshop-on-federated-learning-and-analytics/}
\BIBentrySTDinterwordspacing

\bibitem{abadi2016deep}
M.~Abadi, A.~Chu, I.~Goodfellow, H.~B. McMahan, I.~Mironov, K.~Talwar, and
  L.~Zhang, ``Deep learning with differential privacy,'' in \emph{Proceedings
  of the 2016 ACM SIGSAC conference on computer and communications security},
  2016, pp. 308--318.

\bibitem{2018Hadamard}
J.~Acharya, Z.~Sun, and H.~Zhang, ``Hadamard response: Estimating distributions
  privately, efficiently, and with little communication,'' \emph{AISTATS},
  2018.

\bibitem{as19}
J.~Acharya and Z.~Sun, ``Communication complexity in locally private
  distribution estimation and heavy hitters,'' \emph{CoRR}, vol.
  abs/1905.11888, 2019.

\bibitem{aumuller2021differentially}
M.~Aum{\"u}ller, C.~J. Lebeda, and R.~Pagh, ``Differentially private sparse
  vectors with low error, optimal space, and fast access,'' \emph{arXiv
  preprint arXiv:2106.10068}, 2021.

\bibitem{bagdasaryan2021towards}
E.~Bagdasaryan, P.~Kairouz, S.~Mellem, A.~Gasc{\'o}n, K.~Bonawitz, D.~Estrin,
  and M.~Gruteser, ``Towards sparse federated analytics: Location heatmaps
  under distributed differential privacy with secure aggregation,'' \emph{arXiv
  preprint arXiv:2111.02356}, 2021.

\bibitem{GaussianMech}
B.~Balle and Y.-X. Wang, ``Improving the gaussian mechanism for differential
  privacy: Analytical calibration and optimal denoising,'' in
  \emph{International Conference on Machine Learning}.\hskip 1em plus 0.5em
  minus 0.4em\relax PMLR, 2018, pp. 394--403.

\bibitem{Bassilynips2017}
R.~Bassily, K.~Nissim, U.~Stemmer, and A.~Thakurta, ``Practical locally private
  heavy hitters,'' in \emph{Proceedings of the 31st International Conference on
  Neural Information Processing Systems}, ser. NIPS'17, 2017, p. 2285–2293.

\bibitem{bassily2015local}
R.~Bassily and A.~Smith, ``Local, private, efficient protocols for succinct
  histograms,'' in \emph{Proceedings of the forty-seventh annual ACM symposium
  on Theory of computing}, 2015, pp. 127--135.

\bibitem{bhowmick2019protection}
A.~Bhowmick, J.~Duchi, J.~Freudiger, G.~Kapoor, and R.~Rogers, ``Protection
  against reconstruction and its applications in private federated learning,''
  2019.

\bibitem{bonawitz2017practical}
K.~Bonawitz, V.~Ivanov, B.~Kreuter, A.~Marcedone, H.~B. McMahan, S.~Patel,
  D.~Ramage, A.~Segal, and K.~Seth, ``Practical secure aggregation for
  privacy-preserving machine learning,'' in \emph{proceedings of the 2017 ACM
  SIGSAC Conference on Computer and Communications Security}, 2017, pp.
  1175--1191.

\bibitem{bns19}
M.~Bun, J.~Nelson, and U.~Stemmer, ``Heavy hitters and the structure of local
  privacy,'' in \emph{Proceedings of the 37th {ACM} {SIGMOD-SIGACT-SIGAI}
  Symposium on Principles of Database Systems, Houston, TX, USA, June 10-15,
  2018}.\hskip 1em plus 0.5em minus 0.4em\relax {ACM}, 2018, pp. 435--447.

\bibitem{fc12}
T.-H.~H. Chan, E.~Shi, and D.~Song, ``Privacy-preserving stream aggregation
  with fault tolerance,'' in \emph{Financial Cryptography and Data Security
  (FC)}, 2012.

\bibitem{dperm}
K.~Chaudhuri, C.~Monteleoni, and A.~D. Sarwate, ``Differentially private
  empirical risk minimization,'' \emph{J. Mach. Learn. Res.}, vol.~12, jul
  2011.

\bibitem{cko20}
W.~Chen, P.~Kairouz, and A.~{\"{O}}zg{\"{u}}r, ``Breaking the
  communication-privacy-accuracy trilemma,'' in \emph{Advances in Neural
  Information Processing Systems 33: Annual Conference on Neural Information
  Processing Systems 2020, NeurIPS 2020, December 6-12, 2020, virtual}, 2020.

\bibitem{cormode2018marginal}
G.~Cormode, T.~Kulkarni, and D.~Srivastava, ``Marginal release under local
  differential privacy,'' in \emph{Proceedings of the 2018 International
  Conference on Management of Data}, 2018, pp. 131--146.

\bibitem{cormode2012differentially}
G.~Cormode, C.~Procopiuc, D.~Srivastava, and T.~T. Tran, ``Differentially
  private summaries for sparse data,'' in \emph{Proceedings of the 15th
  International Conference on Database Theory}, 2012, pp. 299--311.

\bibitem{microsoftdp}
B.~Ding, J.~Kulkarni, and S.~Yekhanin, ``Collecting telemetry data privately,''
  \emph{arXiv preprint arXiv:1712.01524}, 2017.

\bibitem{DuchR19}
J.~Duchi and R.~Rogers, ``Lower bounds for locally private estimation via
  communication complexity,'' in \emph{Proceedings of the Thirty-Second
  Conference on Learning Theory}, vol.~99, 2019, pp. 1161--1191.

\bibitem{duchi2018minimax}
J.~C. Duchi, M.~I. Jordan, and M.~J. Wainwright, ``Minimax optimal procedures
  for locally private estimation,'' \emph{Journal of the American Statistical
  Association}, vol. 113, no. 521, pp. 182--201, 2018.

\bibitem{dwork2006calibrating}
C.~Dwork, F.~McSherry, K.~Nissim, and A.~Smith, ``Calibrating noise to
  sensitivity in private data analysis,'' in \emph{Theory of cryptography
  conference}.\hskip 1em plus 0.5em minus 0.4em\relax Springer, 2006, pp.
  265--284.

\bibitem{dworkrothdpbook}
\BIBentryALTinterwordspacing
C.~Dwork and A.~Roth, ``The algorithmic foundations of differential privacy.''
  \emph{Foundations and Trends in Theoretical Computer Science}, vol.~9, no.
  3-4, pp. 211--407, 2014. [Online]. Available:
  \url{http://dblp.uni-trier.de/db/journals/fttcs/fttcs9.html#DworkR14}
\BIBentrySTDinterwordspacing

\bibitem{rappor}
{\'U}.~Erlingsson, V.~Pihur, and A.~Korolova, ``Rappor: Randomized aggregatable
  privacy-preserving ordinal response,'' in \emph{Proceedings of the 2014 ACM
  SIGSAC conference on computer and communications security}, 2014, pp.
  1054--1067.

\bibitem{fanti2015building}
G.~Fanti, V.~Pihur, and {\'U}.~Erlingsson, ``Building a rappor with the
  unknown: Privacy-preserving learning of associations and data dictionaries,''
  \emph{arXiv preprint arXiv:1503.01214}, 2015.

\bibitem{pckv}
X.~Gu, M.~Li, Y.~Cheng, L.~Xiong, and Y.~Cao, ``Pckv: Locally differentially
  private correlated key-value data collection with optimized utility,'' in
  \emph{29th USENIX Security Symposium}, 2020, pp. 967--984.

\bibitem{kairouz2021distributed}
P.~Kairouz, Z.~Liu, and T.~Steinke, ``The distributed discrete gaussian
  mechanism for federated learning with secure aggregation,'' \emph{arXiv
  preprint arXiv:2102.06387}, 2021.

\bibitem{kairouz2015composition}
P.~Kairouz, S.~Oh, and P.~Viswanath, ``The composition theorem for differential
  privacy,'' in \emph{International conference on machine learning}.\hskip 1em
  plus 0.5em minus 0.4em\relax PMLR, 2015, pp. 1376--1385.

\bibitem{ldpdef}
S.~P. Kasiviswanathan, H.~K. Lee, K.~Nissim, S.~Raskhodnikova, and A.~Smith,
  ``What can we learn privately?'' \emph{SIAM Journal on Computing}, vol.~40,
  no.~3, pp. 793--826, 2011.

\bibitem{konevcny2016federated}
J.~Konecny, H.~B. McMahan, F.~X. Yu, P.~Richtarik, A.~T. Suresh, and D.~Bacon,
  ``Federated learning: Strategies for improving communication efficiency,''
  \emph{arXiv preprint arXiv:1610.05492}, 2016.

\bibitem{korolova2009releasing}
A.~Korolova, K.~Kenthapadi, N.~Mishra, and A.~Ntoulas, ``Releasing search
  queries and clicks privately,'' in \emph{Proceedings of the 18th
  international conference on World wide web}, 2009, pp. 171--180.

\bibitem{li2020numeric}
\BIBentryALTinterwordspacing
Z.~Li, T.~Wang, M.~Lopuha\"{a}-Zwakenberg, N.~Li, and B.~Skoric, ``Estimating
  numerical distributions under local differential privacy,'' in
  \emph{Proceedings of the 2020 ACM SIGMOD International Conference on
  Management of Data}, ser. SIGMOD '20.\hskip 1em plus 0.5em minus 0.4em\relax
  New York, NY, USA: Association for Computing Machinery, 2020, p. 621–635.
  [Online]. Available: \url{https://doi.org/10.1145/3318464.3389700}
\BIBentrySTDinterwordspacing

\bibitem{mcmahan2017communication}
B.~McMahan, E.~Moore, D.~Ramage, S.~Hampson, and B.~A. y~Arcas,
  ``Communication-efficient learning of deep networks from decentralized
  data,'' in \emph{Artificial intelligence and statistics}.\hskip 1em plus
  0.5em minus 0.4em\relax PMLR, 2017, pp. 1273--1282.

\bibitem{cdp}
I.~Mironov, O.~Pandey, O.~Reingold, and S.~P. Vadhan, ``Computational
  differential privacy,'' in \emph{Advances in Cryptology - {CRYPTO} 2009, 29th
  Annual International Cryptology Conference, Santa Barbara, CA, USA, August
  16-20, 2009. Proceedings}, ser. Lecture Notes in Computer Science, S.~Halevi,
  Ed., vol. 5677.\hskip 1em plus 0.5em minus 0.4em\relax Springer, 2009, pp.
  126--142.

\bibitem{renting}
\BIBentryALTinterwordspacing
R.~Misra, ``Clothing fit dataset for size recommendation,'' Aug 2018. [Online].
  Available:
  \url{https://www.kaggle.com/rmisra/clothing-fit-dataset-for-size-recommendation}
\BIBentrySTDinterwordspacing

\bibitem{harmony}
T.~T. Nguy{\^e}n, X.~Xiao, Y.~Yang, S.~C. Hui, H.~Shin, and J.~Shin,
  ``Collecting and analyzing data from smart device users with local
  differential privacy,'' \emph{arXiv preprint arXiv:1606.05053}, 2016.

\bibitem{clothing}
\BIBentryALTinterwordspacing
Nicapotato, ``Women's e-commerce clothing reviews,'' Feb 2018. [Online].
  Available:
  \url{https://www.kaggle.com/nicapotato/womens-ecommerce-clothing-reviews}
\BIBentrySTDinterwordspacing

\bibitem{movie}
\BIBentryALTinterwordspacing
Pooh, ``Movie rating data,'' 2017. [Online]. Available:
  \url{https://www.kaggle.com/ashukr/movie-rating-data}
\BIBentrySTDinterwordspacing

\bibitem{qin2016heavy}
Z.~Qin, Y.~Yang, T.~Yu, I.~Khalil, X.~Xiao, and K.~Ren, ``Heavy hitter
  estimation over set-valued data with local differential privacy,'' in
  \emph{Proceedings of the 2016 ACM SIGSAC Conference on Computer and
  Communications Security}, 2016, pp. 192--203.

\bibitem{ndss11}
E.~Shi, T.-H.~H. Chan, E.~Rieffel, R.~Chow, and D.~Song, ``Privacy-preserving
  aggregation of time-series data,'' in \emph{Network and Distributed System
  Security Symposium (NDSS)}, 2011.

\bibitem{shokri2015privacy}
R.~Shokri and V.~Shmatikov, ``Privacy-preserving deep learning,'' in
  \emph{Proceedings of the 22nd ACM SIGSAC conference on computer and
  communications security}, 2015, pp. 1310--1321.

\bibitem{sun2014personalized}
C.~Sun, Y.~Fu, J.~Zhou, and H.~Gao, ``Personalized privacy-preserving frequent
  itemset mining using randomized response,'' \emph{The Scientific World
  Journal}, vol. 2014, 2014.

\bibitem{appledp}
\BIBentryALTinterwordspacing
D.~P. Team, ``Learning with privacy at scale.'' [Online]. Available:
  \url{https://docs-assets.developer.apple.com/ml-research/papers/learning-with-privacy-at-scale.pdf}
\BIBentrySTDinterwordspacing

\bibitem{WangNing2019}
N.~{Wang}, X.~{Xiao}, Y.~{Yang}, J.~{Zhao}, S.~C. {Hui}, H.~{Shin}, J.~{Shin},
  and G.~{Yu}, ``Collecting and analyzing multidimensional data with local
  differential privacy,'' in \emph{2019 IEEE 35th International Conference on
  Data Engineering (ICDE)}, April 2019, pp. 638--649.

\bibitem{wang-miopt-16}
S.~Wang, L.~Huang, P.~Wang, Y.~Nie, H.~Xu, W.~Yang, X.~Li, and C.~Qiao,
  ``Mutual information optimally local private discrete distribution
  estimation,'' \emph{CoRR}, vol. abs/1607.08025, 2016.

\bibitem{wang2017locally}
T.~Wang, J.~Blocki, N.~Li, and S.~Jha, ``Locally differentially private
  protocols for frequency estimation,'' pp. 729--745, 2017.

\bibitem{wang2018locally}
T.~Wang, N.~Li, and S.~Jha, ``Locally differentially private frequent itemset
  mining,'' in \emph{2018 IEEE Symposium on Security and Privacy (SP)}.\hskip
  1em plus 0.5em minus 0.4em\relax IEEE, 2018, pp. 127--143.

\bibitem{randomizedresponse}
S.~L. Warner, ``Randomized response: A survey technique for eliminating evasive
  answer bias,'' \emph{Journal of the American Statistical Association}, 1965.

\bibitem{optimalschemes2018Ye}
M.~Ye and A.~Barg, ``Optimal schemes for discrete distribution estimation under
  locally differential privacy,'' \emph{IEEE Transactions on Information
  Theory}, vol.~64, no.~8, pp. 5662--5676, 2018.

\bibitem{privkv}
Q.~Ye, H.~Hu, X.~Meng, and H.~Zheng, ``Privkv: Key-value data collection with
  local differential privacy,'' in \emph{2019 IEEE Symposium on Security and
  Privacy (SP)}.\hskip 1em plus 0.5em minus 0.4em\relax IEEE, 2019, pp.
  317--331.

\bibitem{privkvm}
Q.~Ye, H.~Hu, X.~Meng, H.~Zheng, K.~Huang, C.~Fang, and J.~Shi, ``Privkvm*:
  Revisiting key-value statistics estimation with local differential privacy,''
  \emph{IEEE Transactions on Dependable and Secure Computing}, 2021.

\bibitem{deeplearncdp}
S.~Yuan, M.~Shen, I.~Mironov, and A.~C.~A. Nascimento, ``Practical, label
  private deep learning training based on secure multiparty computation and
  differential privacy,'' \emph{{IACR} Cryptol. ePrint Arch.}, p. 835, 2021.

\bibitem{zhao2020local}
Y.~Zhao, J.~Zhao, M.~Yang, T.~Wang, N.~Wang, L.~Lyu, D.~Niyato, and K.-Y. Lam,
  ``Local differential privacy-based federated learning for internet of
  things,'' \emph{IEEE Internet of Things Journal}, vol.~8, no.~11, pp.
  8836--8853, 2020.

\end{thebibliography}
